\documentclass{aptpub}
\authornames{NATARAJAN {\it ET AL.}} 
\shorttitle{On a wider class of prior distributions for graphical models} 

\usepackage{enumitem,graphicx,natbib,url,xr-hyper}
\setcitestyle{numbers,square,comma}

\newcommand{\RR}{\mathbb{R}} 

\renewcommand{\Pr}{\mathbb{P}}

\newcommand{\add}[1]{{#1}}  
\newcommand{\addd}[1]{{#1}}  

\begin{document}

\title{On a wider class of prior distributions for\\ graphical models} 



\authorone[\addd{University of Oxford}]{Abhinav Natarajan}

\authortwo[National University of Singapore]{Willem van den Boom}

\authorthree[National University of Singapore]{Kristoforus Bryant Odang}

\authorfour[National University of Singapore; Singapore Institute for Clinical Sciences, A*STAR; University College London]{Maria de Iorio}

\addressone{\addd{Mathematical Institute, University of Oxford, Wellington Square, Oxford OX1 2JD, United Kingdom. Email address: natarajan@maths.ox.ac.uk}}

\addresstwo{Yong Loo Lin School of Medicine, National University of Singapore, 10 Medical Dr, Singapore 117597. Email address: vandenboom@nus.edu.sg}

\addressthree{Yong Loo Lin School of Medicine, National University of Singapore, 10 Medical Dr, Singapore 117597. Email address: kristoforusbryant@u.yale-nus.edu.sg}

\addressfour{Yong Loo Lin School of Medicine, National University of Singapore, 10 Medical Dr, Singapore 117597. Email address: mdi@nus.edu.sg}

\begin{abstract}  
Gaussian graphical models are useful tools for conditional independence structure inference of multivariate random variables.
Unfortunately, Bayesian inference of latent graph structures is challenging due to exponential growth of $\mathcal{G}_n$, the set of all graphs in $n$ vertices.
One approach that has been proposed to tackle this problem is to limit search to subsets of $\mathcal{G}_n$.
In this paper, we study subsets that are vector subspaces with the cycle space $\mathcal{C}_n$ as main example.
We propose a novel prior on $\mathcal{C}_n$ based on linear combinations of cycle basis elements and present its theoretical properties.
Using this prior, we implement a Markov chain Monte Carlo algorithm\add{,} and show that (i) posterior edge inclusion estimates \add{computed with our technique} are comparable \add{to estimates from the standard technique} despite searching a smaller graph space\add{,} and (ii)
the vector space perspective enables straightforward \add{implementation of} MCMC algorithms.
\end{abstract}



\section{Introduction}
\label{sec1}

Gaussian graphical models (GGMs) \cite{dempster1972covariance, lauritzen1996graphical} have become a standard technique to represent the conditional independence structure of a set of random variables.
Given an undirected graph $G = (V, E)$, the set of vertices $V=\{1,\ldots,n\}$ represents the random variables while the set of edges $E \subseteq \{(i,j) \in V \times V: i < j\}$ represents conditional dependencies between the variables.
In the Bayesian literature, a latent graph $G$ is often inferred by first specifying a prior $p(G)$ on the space of graphs followed by a prior on the precision matrix $K$ conditional on the graph, $p(K \mid G)$. 
Then, posterior inference is performed through sampling algorithms such as Markov chain Monte Carlo (MCMC) \cite[e.g.][]{giudici1999decomposable, jones2005experiments} or sequential Monte Carlo (SMC) \cite{tan2017bayesian, van2021unbiased}.
One of the main challenges in graph inference is that the number of graphs grows exponentially with the number of vertices since the size of the space $\mathcal{G}_n$ of all  graphs with $n$ vertices is $2^{n (n -1) / 2}$.

There are two main lines of research addressing the computational difficulties associated with the size of graph space: (i) devising efficient sampling algorithms and (ii) considering particular subsets of $\mathcal{G}_n$ with desirable properties.
This work explores novel subsets in light of this second direction.
A particularly popular restriction of $\mathcal{G}_n$ consists of focusing on the set of decomposable graphs as it allows for exact posterior computation under the conjugate Hyper-Inverse Wishart prior (see, for example, \cite{giudici1996learning,giudici1999decomposable, green2013sampling,lauritzen1996graphical, scott2008feature, wang2010simulation}). 
However, the decomposability assumption is often restrictive in applications:
suppose that the ``true'' graph is non-decomposable.  Niu et al.~\cite{niu2021bayesian} prove that, as the number of observations goes to infinity, the posterior distribution concentrates at the minimal triangulations of the true graph. These triangulations can present $O(n \log n)$ spurious edges \citep{chung1994chordal}.

As an alternative to decomposable graphs, recent research effort has been devoted to consider 
the set of spanning trees $\mathcal{T}_n$ on the complete graph (i.e.\ the graph containing all possible edges), which is necessarily a subset of decomposable graphs as trees do not contain cycles.
Højsgaard et al.~\cite{hojsgaard2012graphical} propose a maximum a posteriori estimator of the latent spanning tree based on Kruskal's \add{\cite{kruskal1956shortest}} algorithm.
Schwaller et al.~\cite{Schwaller2019} are able to derive the exact posterior distribution of each tree in $\mathcal{T}_n$ given a set of observations by applying Kirchoff's matrix tree theorem \citep{chaiken1978matrix} while Duan and Dunson~\cite{duan2021bayesian} propose a Bayesian model for spanning trees where the  likelihood function is built on a generative graph process and leads to an efficient Gibbs sampling algorithm. Although spanning trees provide computational advantages compared to considering all graphs,
they only allow inference of the ``backbone'' of the graph \citep{duan2021bayesian}, i.e.\ the minimum spanning tree of the data generating graph, at the cost of losing information about its global structures.

In this paper, we investigate larger subsets of $\mathcal{G}_n$ arising from an algebraic perspective on graphs with the cycle space as the main concrete example.
Firstly, note that the set $\mathcal{G}_n$ together with the edgewise \emph{modulo two addition} $\oplus$ and the \emph{trivial multiplication} (both operations are defined in Section~\ref{sec: cycle-basis}) forms a vector space on the finite field $\mathbb{Z}_2$, where the set of all edges in the complete graph forms the standard basis of $\mathcal{G}_n$.
Definitions and preliminary results about the graph vector space are deferred to Section~\ref{sec2}.

Our work builds on the fact
that alternative bases for $\mathcal{G}_n$ exist.
In particular, we consider the example of the space spanned by linear combinations of 
cycles on $V$
which forms the cycle space $\mathcal{C}_n$. 
The cycle space is a proper subspace of $\mathcal{G}_n$ which can be spanned using bases consisting of cycles \citep{wallis2010beginner}.
We investigate the  theoretical properties of \add{this} space
and of the prior distribution on graphs
when the prior is specified through a prior on cycles, instead of the common practice of specifying a prior on the edges,
as well as the implications for statistical inference of such a prior.
We show that the cycle space $\mathcal{C}_n$ is a substantially smaller set than $\mathcal{G}_n$,
though this does not notably affect posterior inference
such as edge inclusion probabilities in the concluding application.

The paper is structured as follows. Section~\ref{sec2} reviews graphical modelling concepts, how $\mathcal{G}_n$ constitutes a vector space, cycle basis theory and compares the cycle space with the set of decomposable graphs. 
Section~\ref{sec:theory} analyses the cycle basis prior.
\addd{Section~\ref{sec:simul} contains simulation studies.}
In 
Section~\ref{sec5}, we apply the proposed algorithm to gene expression data. 
Section~\ref{sec6} concludes the paper and discusses potential future directions.
All proofs can be found in Supplementary Material.

\section{Background} 
\label{sec2}

\subsection{Graph preliminaries}
\label{sec:graph_prelim}

We consider undirected graphs with no loops and no multi-edges, also called simple graphs.
A \emph{path} $\gamma_{\alpha,\beta}$
between two vertices $\alpha,\beta\in V$ is
a sequence of edges $(e_1, \dots, e_m)$
such that $v_1 = \alpha$, $v_{m+1} = \beta$ and $e_i = (v_i, v_{i+1})$ for all $i \in \{1, \dots m\}$
where every vertex in the vertex sequence $(v_1, \dots, v_{m+1})$ is distinct.
\add{A \emph{circuit} is a
sequence of edges $(e_1, \dots, e_m)$
where $v_1 = v_{m+1}$.}
A \emph{cycle} is a \add{circuit}
where every vertex in $(v_1, \dots, v_m)$ is distinct.
A graph is \emph{connected} when a path $\gamma_{\alpha,\beta}$ exists for all $\alpha, \beta \in V$.
A graph is \emph{complete} if it contains all possible edges, that is $E = \{(i,j) \in V \times V: i < j\}$.
A \emph{tree} is a graph where there is exactly one path between any pair of vertices. A \emph{spanning tree} of a connected graph $G = (V, E)$ is a tree with vertex set $V$ and an edge set which is a subset of $E$.
\add{A star tree $T$ on the vertices $\{v_0, \ldots, v_{n-1}\}$ is the spanning tree whose edges are $(v_0, v_i)$ for $i = 1, \ldots, n-1$. In this case, $v_0$ is the root of $T$.}
The degree of a vertex $v$ is the number of edges incident to $v$.

\subsection{Bayesian inference of GGMs}

In a GGM, a graph $G=(V,E)$ on $n$ vertices models the zeros in the precision matrix $K$
of the Gaussian distribution $N(0, K^{-1})$ on the independent rows of the $N\times n$ data matrix $X$:
\begin{equation} \label{eq: x-given-K}
p(X\mid K) \propto |K|^{N/2} \exp(-\langle K, U\rangle /2 )
\end{equation}
where
$U = X^T X$ and $\langle A, B\rangle = \mathrm{tr}(A^T B)$.
Specifically, $K_{ij}=0$ if \add{$(i,j) \notin E$}.
For Bayesian inference,
Dawid and Lauritzen~\cite{dawid1993hyper} introduce the Hyper-Inverse Wishart prior which is a conjugate prior $p(K \mid G)$ for decomposable graphs $G$.
\add{Giudici~\cite{giudici1996learning} and}
Roverato~\cite{roverato2002hyper} generalise this idea to non-decomposable graphs, proposing  the $G$-Wishart distribution $W_G(\delta, D)$ with degrees of freedom $\delta>2$ and rate matrix $D$ which has density 
\begin{equation} \label{eq: K-given-G}
p(K\mid G) = \frac{1}{I_G(\delta, D)} |K|^{(\delta -2)/2}  \exp(-\langle K, D\rangle /2) 
\end{equation}
with respect to the Lebesgue measure on the set of positive definite matrices with zeros imposed by $G$.
Here, $I_G(\delta, D)$ is a normalising constant.

Combining the prior in \eqref{eq: K-given-G} and the likelihood \eqref{eq: x-given-K},
$K\mid G,X \sim W_G (\delta + N, D + U)$
due to conjugacy.
The likelihood for $G$ with $K$ marginalised out follows as \citep[e.g.][]{atay2005monte}
\begin{equation}\label{eq: marginal-lik}
p(X \mid G) = \frac{I_G(\delta + N, D + U)}{(2\pi)^{Nn/2} I_G(\delta, D)}
\end{equation}

Unfortunately, computing $I_G(\delta, D)$ for a general graph $G$, despite a recently derived recursive expression \citep{uhler2018exact}, is intractable.
Monte Carlo \citep{atay2005monte} and Laplace \citep{lenkoski2011computational} approximations have thus been suggested, along with methods that avoid the normalising constant computation entirely \add{\citep{wang2012efficient}} using the exchange algorithm of Murray et al.~\cite{murray2012mcmc}.
Notably,
$I_G(\delta, D)$
is computationally tractable for decomposable $G$,
motivating
the common restriction of
the space of graphs $\mathcal{G}_n$
to decomposable graphs.

Having the marginal likelihood of the observations given a general graph,  Bayesian inference is performed by defining a prior $p(G)$ on $\mathcal{G}_n$.
Such prior distributions are usually constructed through the specification of edge inclusion probabilities.
In the next section, we define a prior distribution over graph space  based on more complex substructures than edges. 

\subsection{$\mathcal{G}_n$ as vector space and the cycle space} \label{sec: cycle-basis}

The set $\mathcal{G}_n$ of all graphs with $n$ vertices is part of a vector space over the finite field $\mathbb{Z}_2$.
Here, $\mathbb{Z}_2$ is a field on the set $\{0, 1\}$ where addition is the modulo 2 addition (i.e.\ $0 + 0 = 0$, $1 + 0 = 0 + 1 = 1$, and $1 + 1 = 0$) while multiplication is trivial (i.e.\ $0 \times x  = 0$ and $1 \times x = x$ for any $x \in \{0, 1\}$).
The vector space is the triple $(\mathcal{G}_n, \oplus, \otimes)$ where $\oplus$ is the modulo 2 addition of the edges and $\otimes$ is the (trivial) multiplication such that $0 \otimes G$ is the graph with no edges, and $1 \otimes G = G$.
Note that the set of all edges in the complete graph forms one basis of this vector space which we call the \emph{standard basis}. 
Alternative bases exist for the same vector space, which can (i) be used as building components of a graph (in a similar way as with edges) and (ii)
induce priors $p(G)$ through the specification of a distribution
on the elements of the basis.
For illustration,
Figure~\ref{fig:alternative-basis} considers three different bases of $\mathcal{G}_n$ with $n=4$ vertices.
\begin{figure}[htb]
    \centerline{\includegraphics[width=\textwidth]{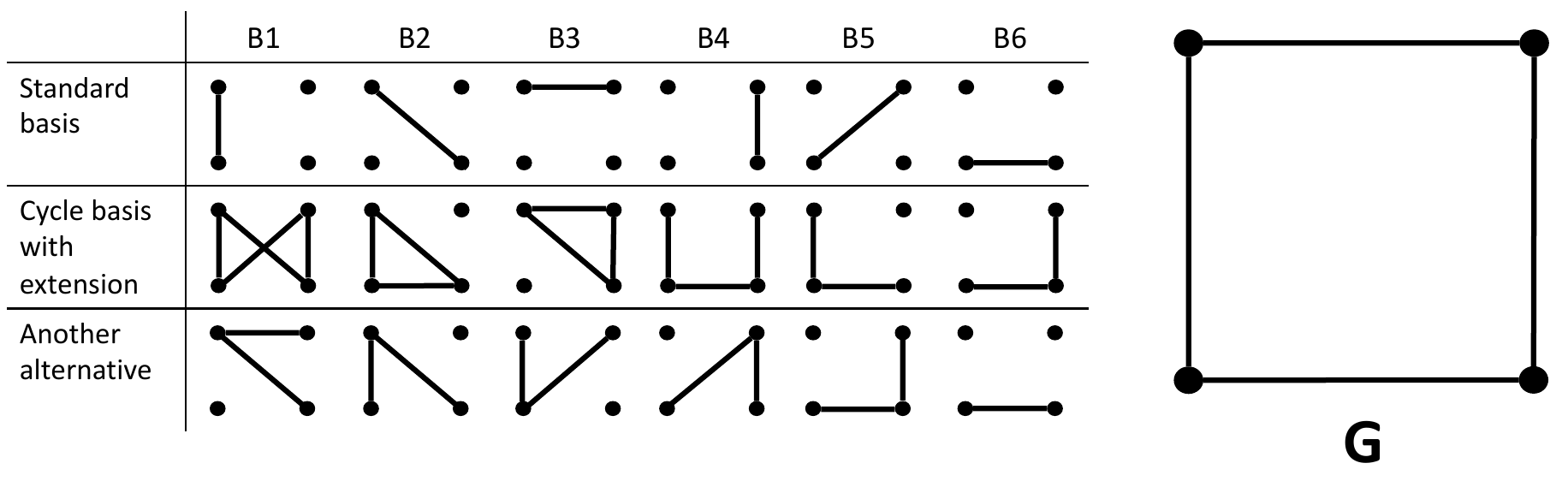}} 
	\caption{The table shows three different bases of the graph space $\mathcal{G}_n$ with $n=4$ vertices. Consider the decomposition of the graph $G\in \mathcal{G}_n$ on the right
	over each of the bases:
	with the standard basis, $G = \textnormal{B1} + \textnormal{B3} + \textnormal{B4} + \textnormal{B6}$.
	With the cycle basis, $G = \textnormal{B2} + \textnormal{B3}$.
	For the last basis, $G = \textnormal{B1} + \textnormal{B2} + \textnormal{B5}$.}
	\label{fig:alternative-basis}
\end{figure}

The basis of this work is restricting the graph space in GGMs by considering subspaces of $\mathcal{G}_n$.
A \emph{subspace} of the graph vector space is a subset of $\mathcal{G}_n$ that is closed under $\oplus$.
Restricting graphs via an appropriate subspace is desirable in GGMs for two reasons.
Firstly, it shrinks the search space to a subset of $\mathcal{G}_n$.
Secondly, being closed under addition, a subspace lends itself well to convenient Monte Carlo sampling steps as algorithms that only propose states in the subspace are readily constructed.
In this paper, we focus on one particular subspace, namely the cycle space:
\begin{definition}[Cycle space]\label{dfn: cycle-space}
The cycle space $\mathcal{C}_n$ of graphs with $n$ vertices is the set of linear combinations of cycles in $\mathcal{G}_n$. 
\end{definition} 

We now present some relevant properties of $\mathcal{C}_n$.
By its definition, $\mathcal{C}_n$ is closed under addition making it a proper vector space:
\begin{enumerate}[label=P\arabic*]
    \item The cycle space $\mathcal{C}_n$ is a subspace of the vector space $(\mathcal{G}_n, \oplus, \otimes)$ \citep[Theorem~5.1]{wallis2010beginner}.
\end{enumerate}
Whether a graph $G\in\mathcal{G}_\add{n}$ is in $\mathcal{C}_n$
can be conveniently checked using Veblen's theorem:
\begin{enumerate}[label=P\arabic*,resume]
    \item \label{p:veblen}
    A graph $G\in\mathcal{G}_n$ is an element of $\mathcal{C}_n$ if and only if every vertex has even degree in $G$, i.e.\ every vertex has an even number of neighbours \citep{Veblen1912}.
\end{enumerate}
\add{
A more intuitive description of the graphs in $\mathcal{C}_n$
is that they are precisely those that are the union of edge-disjoint cycles \citep[Theorem~5.1]{wallis2010beginner}.
Also, each connected component of a graph in $\mathcal{C}_n$
is a circuit.
Notably,
the edge union of cycles with overlapping edges is not necessarily in $\mathcal{C}_n$:
the edge union of the cycles $(1,2,3,1)$ and $(1,2,4,1)$ in the example from Figure~\ref{fig:alternative-basis} results in a degree of three for vertices~1 and 2.
Also, whether the complete graph is in $\mathcal{C}_n$ depends on the parity of $n$ by Property~\ref{p:veblen}.
}

Bases of $\mathcal{C}_n$ can be readily found using fundamental systems of cycles.
\begin{definition}[Fundamental system of cycles]\label{defn: system-cycles} 
Let $T = (V, E_T)$ be a spanning tree of
the complete graph
and $\overline{T} = (V, E_{\overline{T}})$ be the complement of $T$.
That is, $e\in E_{\overline{T}}$ if and only if $e\notin E_{T}$.
The \emph{fundamental system of cycles} with respect to $T$ is the set of graphs obtained by taking each cycle formed by adding one edge in $\overline{T}$ to $T$.  
\end{definition}
(We constrain ourselves to \addd{spanning trees of} the complete graph for simplicity, though note that the fundamental system of cycles is usually also defined for incomplete graphs.)
\begin{enumerate}[label=P\arabic*,resume]
\item  \label{p:bases}
 Every fundamental system of cycles is a basis of $\mathcal{C}_n$ \citep[Theorem~5.5]{wallis2010beginner}.
\end{enumerate}
\begin{figure}[htb]
    \centerline{\includegraphics[width=\textwidth]{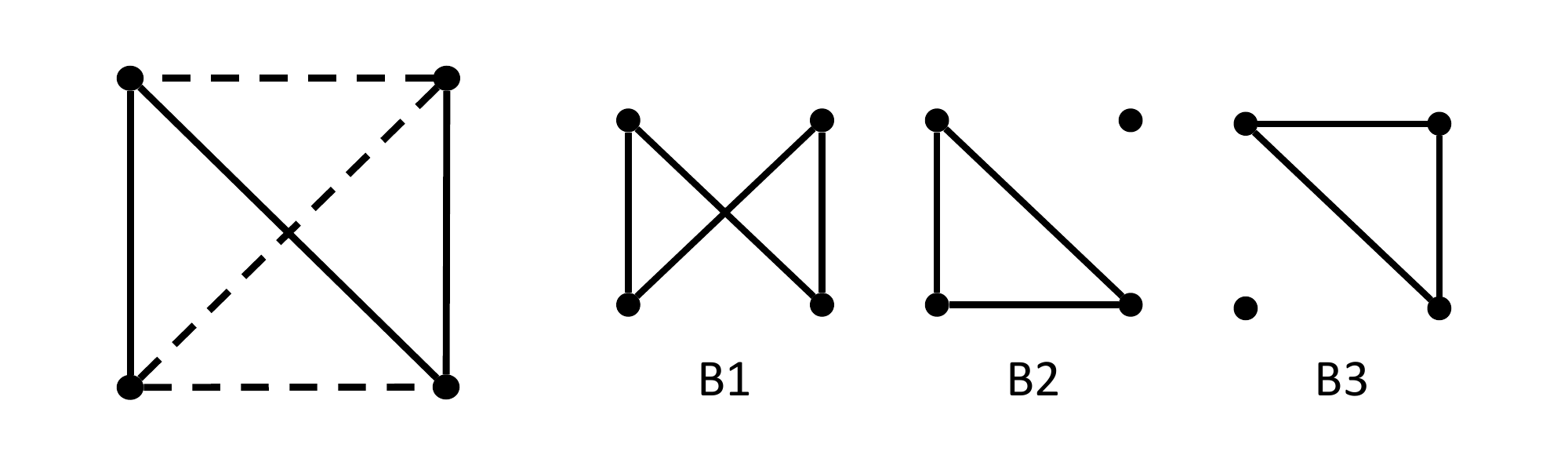}} 
	\caption{Illustration of a fundamental system of cycles. The graph on the left visualises the spanning tree $T$ with solid lines and its complement $\overline{T}$ with dashed lines. The basis elements on the right are the cycles obtained by adding one of the dashed edges in $\overline{T}$ to $T$.}
	\label{fig: cycle-space-from-tree}
\end{figure}
Figure~\ref{fig: cycle-space-from-tree}
visualises how to obtain the fundamental system of cycles that constitutes the cycle basis considered in Figure~\ref{fig:alternative-basis}.
In Figure~\ref{fig:alternative-basis},
this basis, which spans $\mathcal{C}_n$, is extended to span the whole of $\mathcal{G}_n$.
Such extensions exist for any basis spanning a subspace
\citep{sheldon2015linear}.
Since $T$ has $(n - 1)$ edges
and the number of elements in the fundamental system of cycles equals the number of edges in $\overline{T}$,
Property~\ref{p:bases} implies the dimension of $\mathcal{C}_n$:
\begin{enumerate}[label=P\arabic*,resume]
\item \label{p:dim}
The number of elements in a cycle basis and thus the dimension of the vector space $\mathcal{C}_n$ is $n(n-1)/2 - (n - 1) = (n-1)(n-2)/2$.
\end{enumerate}

\add{

As considered in Figure~\ref{fig:alternative-basis},
the number of basis elements in a decomposition of a graph varies with the basis considered.
That same is true when only considering cycle bases.
Consider for instance Figure~\ref{fig:n5_example}.
The graph that is basis element B6 for the path graph
involves not one but three basis elements for the cycle basis derived from the star tree (B1 + B2 + B5).

\begin{figure}[tb]
    \centerline{\includegraphics[width=\textwidth]{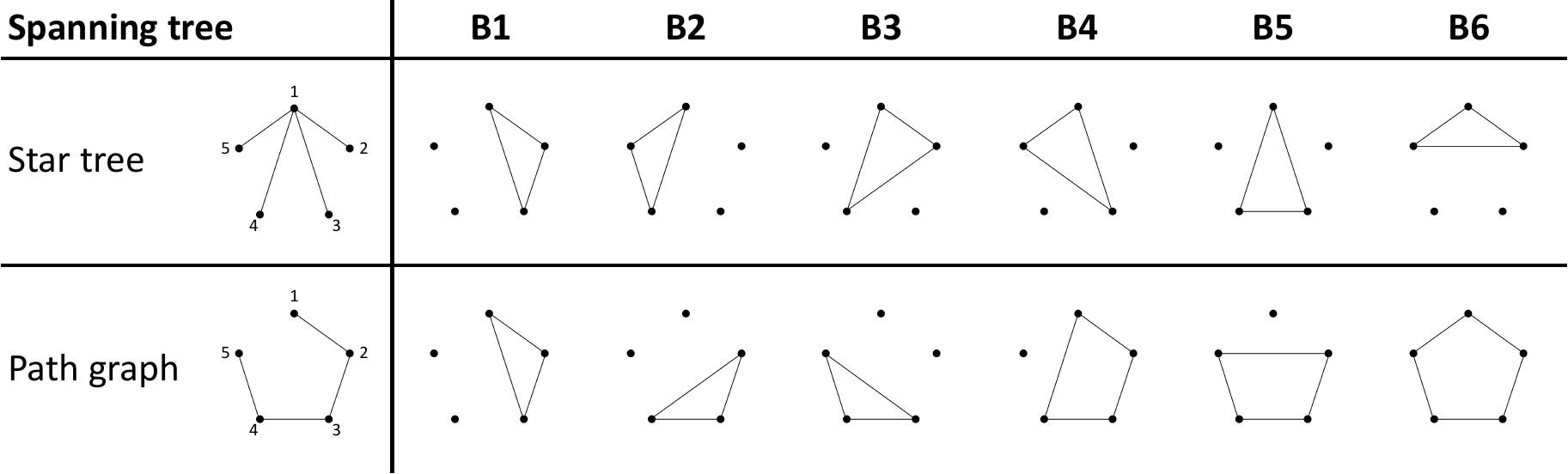}} 
	\caption{\add{Illustration of two fundamental systems of cycles with $n=5$ vertices.}}
	\label{fig:n5_example}
\end{figure}

}

\subsection{Comparing the cycle space with the set of decomposable graphs}
\label{sec:comparison}

As discussed in Section~\ref{sec1}, there have been efforts to limit the search space to a subset of $\mathcal{G}_n$, most notably to the set of decomposable graphs.
While this restriction has the advantage of allowing normalising factors to be computed exactly, there are drawbacks to this restriction especially in applications.
Firstly, unlike the cycle space, the set of decomposable graphs is not known to be closed under any vector addition operation.
\add{
As a result,
more involved
MCMC algorithms
have been proposed to ensure that proposed graphs are decomposable
via
constraints on which edge to flip
\citep{giudici1999decomposable,Stingo2015}
or by sampling from the larger set of junction trees, which correspond to decomposable graphs 
\citep{green2013sampling}.}
On the other hand, in the cycle space, we are able to propose moves which are guaranteed to be in $\mathcal{C}_n$ by standard properties of vector spaces.
\add{Relatedly, decomposability cannot be checked as easily as Property~\ref{p:veblen} which determines membership of $\mathcal{C}_n$.}

Another,
more important drawback of restricting inference to decomposable graphs is that, in general, decomposable graphs do not approximate other graphs well. For instance,
$O(n \log n)$ edges
need to be added to a graph
to make the resulting graph decomposable in general
\citep{chung1994chordal}.
This contrast with the cycle space
as  any graph $G\in\mathcal{G}_n$
is \emph{at most} $n/2$ edges different from
a graph in $\mathcal{C}_n$
per the following result:

\begin{proposition} \label{thm:dense}
Let $k$ be the number of
odd-degree vertices
in $G\in\mathcal{G}_n$.
Then, there exists a graph in $ \mathcal{C}_n$ that differs from $G$ by
$k/2$ edges
while there exists no graph in $ \mathcal{C}_n$ that differs from $G$ by less than $k/2$ edges.
\end{proposition}

The size of $\mathcal{C}_n$ is also larger than that of the set of decomposable graphs.
Note that $\mathcal{G}_n$ has $2^{n (n-1)/ 2}$ elements since the standard basis is the edge set containing $n (n-1)/ 2$ elements.
By Property~\ref{p:dim},
the cardinality of $\mathcal{C}_n$ is $2^{(n-1) (n-2) / 2}$.
In contrast, the number of decomposable graphs on $n$ vertices tends towards $2^{n^2/4 + O(n \log n)}$ as $n\to\infty$
\citep{bender1985almost}.
(For comparison, the number of spanning trees on the complete graph, which are decomposable, is $|\mathcal{T}_n|=n^{n-2}$ by Cayley's~\cite{Cayley1889} formula.)
Thus, constraining graphs to the cycle space is substantially less restrictive than assuming decomposability with a clear upper bound to the difference between any true graph and $\mathcal{C}_n$.
\add{Lastly, the set of decomposable graphs is not a subset of $\mathcal{C}_n$. For instance, trees are not in $\mathcal{C}_n$ by Property~\ref{p:veblen} since their leaf vertices have degree one.
}

\section{Theoretical properties of cycle basis priors}\label{sec:theory}

A prior distribution on $\mathcal{C}_n$ can be induced by assigning a distribution to cycle basis elements. In this section, we explore the theoretical properties of such priors.
Here, $C$ denotes a basis of cycles for $\mathcal{C}_n$.
Thus, $C$ is a set of $(n-1)(n-2)/2$
cycles.
The results hold for any fixed basis $C$ unless otherwise noted.
We begin by showing that it is possible to induce the uniform prior on $\mathcal{C}_n$.

A cycle basis prior can be defined in terms of cycle-inclusion probabilities. 
Similar\add{ly} to the standard basis, when the cycle-inclusion probability is 0.5, we get the uniform prior on $\mathcal{C}_n$.
\begin{proposition} \label{prop:uniform}
\add{Let $C$ be any cycle basis for the cycle space $\mathcal{C}_n$.} Suppose that the inclusion of each cycle in $C$ is an independent Bernoulli random variable with probability $0.5$. Then, the induced distribution on
$\mathcal{C}_n$ is uniform \add{and the marginal edge inclusion probabilities equal $0.5$}.
\end{proposition}
\add{
This uniformity holds for any $C$ and thus also for the distribution on $\mathcal{C}_n$ induced by a distribution on such bases $C$.
The edge inclusions are not independent under the uniform distribution over $\mathcal{C}_n$:
independent edges with probability $0.5$ yields the uniform distribution over $\mathcal{G}_n$.
Instead, edge flips can only happen jointly on sets of edges that correspond to graphs in $\mathcal{C}_n$ to stay in the cycle space, i.e.\ to continue to satisfy Property~\ref{p:veblen}.

Often, there is interest in non-uniform distributions, for instance to induce sparsity.}
In general, the edge inclusion probabilities induced by inclusion probabilities of cycle basis elements depend on the choice of cycle basis, and hence on the choice of spanning tree used to generate this basis via the fundamental system of cycles (Property~\ref{p:bases}).
Although a closed-form expression is not available, we propose an efficient algorithm to compute the edge inclusion probabilities from a cycle basis prior.

\begin{proposition} \label{prop:edge_prop}
Let $e \in \{(i,j) \in V \times V: i < j\}$ be any edge. Let $\{c_1, \ldots, c_r\}\subset C$ be the set of basis elements that contain $e$. Suppose that the inclusions of these elements are independent \add{with inclusion probabilities $\{p_1, \ldots, p_r\}$}. Define the polynomial
\[
f(x) = \prod_{i=1}^r (1-p_i + p_i x).
\]
Then, the induced marginal probability of inclusion of the edge $e$ is the sum of the coefficients of the odd powers of $x$ in $f(x)$.
\add{This edge probability reduces to
${\{1 - (1-2p)^r\}/2}$
if the probabilities $\{p_1, \ldots, p_r\}$ are all equal to $p$.}
\end{proposition}
\add{If the cycle basis is generated from a star tree with independent and equally probable inclusion of basis elements, then these probabilities simplify.
(The fundamental system of cycles with respect to a star tree is the set of all cycles of three edges that are incident to \add{its root}.)
\begin{corollary} \label{cor:edge_prob_star}
Let the basis $C$ be
defined as the fundamental system of cycles with respect to a star tree $T$ on $n\geq 2$ vertices.
Suppose that the cycle inclusions are independent Bernoulli random variables  with probability $p$.
Then, the marginal edge inclusion probabilities are given by
\[
    \Pr\{(i,j) \in E \mid T\} =
\begin{cases}
    \{1-(1-2p)^{n-2}\}/2, &\text{$i$ or $j$ is the root of $T$} \\
    p, &\text{otherwise}
\end{cases}
\]
Moreover, the distribution
induced by the uniform distribution over all star trees
has $\Pr\{(i,j) \in E\} = p + \{1-2p - (1-2p)^{n-2}\}/n$.
\end{corollary}}

Polynomial multiplication can be performed efficiently using linear convolution based on a fast Fourier transform (FFT), making the computation of the probabilities \add{in Propostion~\ref{prop:edge_prop}} efficient. For the sake of completeness, we also provide an algorithm to compute the joint distribution of edge inclusions induced by the cycle basis prior. This is useful in, for example, calculating the degree distribution of a vertex.

\begin{proposition} \label{prop:edge_prop_joint}
Let $v\in V$ be any vertex.
Let $\{e_1, \ldots, e_m\}$ be the set of edges incident to $v$.
Let $\{c_1, \ldots, c_r\} \subset C$ be the set of cycles incident to $v$.
Suppose that the cycle inclusions are independent \add{with inclusion probabilities $\{p_1, \ldots, p_r\}$}.
For $i=1,\dots,r$, let $a(i), b(i) \in \{1, \ldots, m\}$ be such that $e_{a(i)}, e_{b(i)}$ are the edges of $c_i$ that are incident to $v$. Define the polynomial
\[
f(t_1, \ldots, t_r) = \prod_{i=1}^r (1-p_i + p_it_i).
\]
Let $g$ be the image of $f$ in the polynomial ring $\RR[x_1, \ldots, x_m]/ \langle x_1^2, \ldots, x_m^2 \rangle$ under the unique ring homomorphism satisfying $t_i \mapsto x_{a(i)}x_{b(i)}$. Then, the probability of including any set of edges $\{e_s\}_{s\in S}$ while excluding the other $m-|S|$ edges incident to $v$ is given by the coefficient of $\prod_{s\in S} x_s$ in $g$.
\end{proposition}

The above proposition can be applied in practice by noting that $g$ is equal to the product of the polynomials $1 - p_i + p_i x_{a(i)}x_{b(i)}$ in the ring $\RR[x_1, \ldots, x_m]/\langle x_1^2, \ldots, x_m^2 \rangle$. Multiplication of linear polynomials modulo squares is equivalent to a circular convolution. By the circulant convolution theorem, this can be efficiently performed using an FFT of length 2 in $m$ dimensions.

Although the above algorithm can be used to calculate degree distributions in the general case, closed-form expressions may be computed when the spanning tree used to generate the cycle basis has a tractable structure as in the case of star trees.

\begin{proposition} \label{prop:deg_dist}
Let the basis $C$ be
defined as the fundamental system of cycles with respect to the star tree on the vertices $\{v_0, \ldots, v_{n-1}\}$, $n\geq 2$, rooted at $v_0$.
Suppose that the cycle inclusions are independent Bernoulli random variables  with probability $p$.
Then, the degree distribution for the vertices $\{v_1, \dots, v_{n-1}\}$ is given by
\[
    \Pr\{\deg(v_i) = k\} =
\begin{cases}
    (1-p)^{n-2}, &k=0 \\ \sum_{j=k-1}^k\binom{n-2}{j}p^j (1-p)^{n-2-j}, &\text{$2\leq k< n$ and $k$ is even} \\
    0, &\text{otherwise}
\end{cases}
\]
\end{proposition}

We now briefly discuss the sparsity of graphs in the cycle space.
\add{For a given probability $p$ of independent basis element inclusions,
Proposition~\ref{prop:edge_prop}
gives the
edge probability $\{1 - (1-2p)^r\}/2$ which is an increasing function of $p$ for $p\leq 0.5$
and which involves the number $r$ of elements in the basis $C$ that include the edge.
Thus, setting a smaller $p$ yields shrinkage on the number of edges: the edge probability can be made arbitrarily small via $p$.
This shrinkage depends on $C$ via $r$.
For instance,
consider Figure~\ref{fig:n5_example} and edge $(2,3)$.
There, $r=1$ for $C$ generated by the star tree and $r=5$ for $C$ generated by the path graph.
In general,
$r=1$ for any edges in the complement $\overline{T}$
of the spanning tree $T$ that generated $C$.
For edges in $T$, $r$ varies with $C$.

More generally},
there is no simple relationship between the sparsity of a graph in terms of its cycle basis inclusions and the edge-sparsity of that graph. However, analytical results are available in the particular case that the cycle basis is generated from a star tree.
\begin{proposition} \label{prop: shrinkage}
Let the basis $C$ be
defined as the fundamental system of cycles with respect to a star tree on $n$ vertices.
Define $m=\lfloor(n-1)/2\rfloor$.
Consider the graph $G = (V, E)$ formed by the inclusion of $\add{q}$ basis elements from $C$.
Then, the number of edges $|E|$ is bounded as $\add{q} \leq |E| \leq \add{q} + 2 \min(\add{q},\, m)$.
The upper bound is tight for $\add{q}\leq m$.
\end{proposition}
This implies that specifying a shrinkage prior on $\add{q}$ leads to shrinkage on the number of included edges for cycle bases induced by a star tree.
We also consider $|E|$ as a function of $\add{q}$ empirically in Figure~\ref{fig:edge_bound}.
The empirical results show that a number $\add{q}$ implies a smaller range of $|E|$ than suggested by Proposition~\ref{prop: shrinkage}, especially for a large number of vertices.
\begin{figure}[tb]
    \centerline{\includegraphics[width=\textwidth]{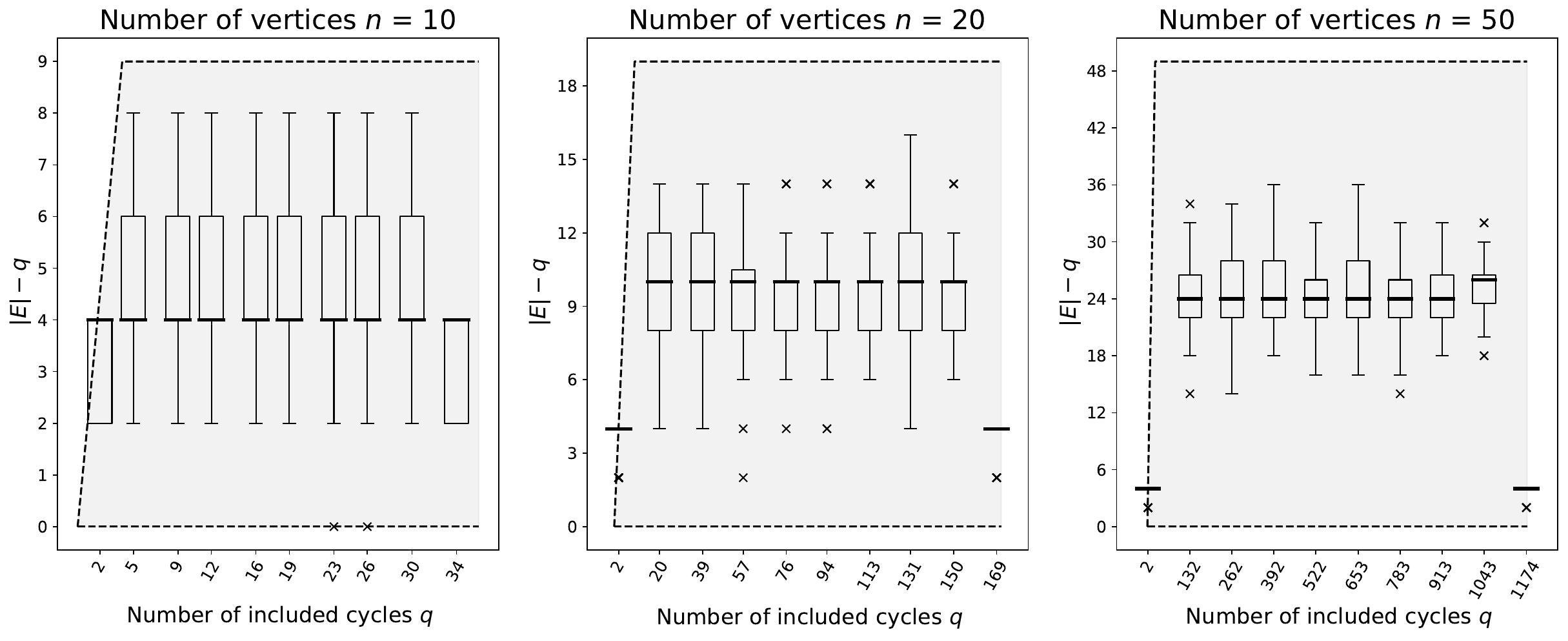}}
	\caption{Box plots of $|E|-\add{q}$ from 100 graphs $G$ generated by sampling $\add{q}$ elements from a basis $C$ generated by a star tree without replacement.
	The dashed lines mark the bounds $0 \leq |E|-\add{q} \leq 2 \min(\add{q},\, m)$ from Proposition~\ref{prop: shrinkage}.}
	\label{fig:edge_bound}
\end{figure}
\add{
Simulation studies in Section~\addd{\ref{sec:simul}} confirm that limiting $q$ also results in lower posterior edge inclusion probabilities, as does lowering the prior basis element inclusion probability $p$.
}

Apart from the prior process on cycle bases, we also consider the prior induced by edge unions of spanning trees in Supplementary Material.
However, we find that computation of
the induced
prior $p(G)$ is intractable since edge unions of spanning trees do not form a vector space.
We propose an algorithm to compute $p(G)$\add{, specifically to count the number of decompositions of a graph into spanning trees,} with
complexity $O(2^{|E|}\, n^3)$ which is an improvement over the
superexponential complexity of the
brute-force enumeration method. 
Further, we prove bounds to attempt an approximation of the prior ratio as appearing in a Metropolis-Hastings acceptance probability, but simulations shows that the bounds are too wide be of use (see Supplementary Material).

\addd{\section{Simulation studies with the cycle basis prior}
\label{sec:simul}

We conduct simulation studies to better understand what the effect of certain cycle basis priors on posterior inference is.
We do not apply these priors to the gene expression application in Section~\ref{sec5} as the decomposition of a graph into changing cycle bases in Step~\ref{supp:step:decomposition} of Algorithm~\ref{supp:alg:MCMC} in Supplementary Material is too expensive with $n=93$ vertices.
We simulate the $N\times n$ data matrix $X$ from the model $p(X\mid K)$ in \eqref{eq: x-given-K}
with the $n\times n$ precision matrix $K$
given by
$K_{ii}=1$ for $i=1,\dots,n$
and
$K_{12}=K_{13}=K_{23}=K_{34}=K_{35}=K_{45}=0.3$
while all other upper-triangular elements of $K$ equal zero.
This fully defines $K$ as it is a symmetric matrix.
Thus, the true underlying graph $G$ corresponds to the union of the two cycles $(1,2,3)$ and $(3,4,5)$.
We set $n=15$ as number of vertices and $N=150$ as number of observations.

We consider six different priors on graphs.
The first one is the edge basis prior with independent edge inclusions with probability $p=0.5$.
The next four priors are constrained to the cycle space $\mathcal{C}_n$ as follows.
The second prior is induced by the uniform prior $p(T)$ over all spanning trees
combined with the prior $p(G\mid T)$ resulting from a priori independent basis element inclusions with probability $p=0.5$ where the cycle basis is generated by spanning tree $T$.
The third prior uses the same uniform $p(G\mid T)$
but with $p(T)$ uniform over all star trees instead of the larger set of all spanning trees.
The fourth prior is the same as the third prior but with a different $p(G\mid T)$ where the a priori
basis element inclusion probability is $p=0.05$ instead of $0.5$.
The fifth prior is like the third prior but with $p(G\mid T)$ uniform over all graphs that consist of $q=3$ or fewer elements from the cycle basis generated by $T$.
Finally, the last prior is uniform over all decomposable graphs.

\begin{figure}[tbp]\add{
    \centerline{\includegraphics[width=\textwidth]{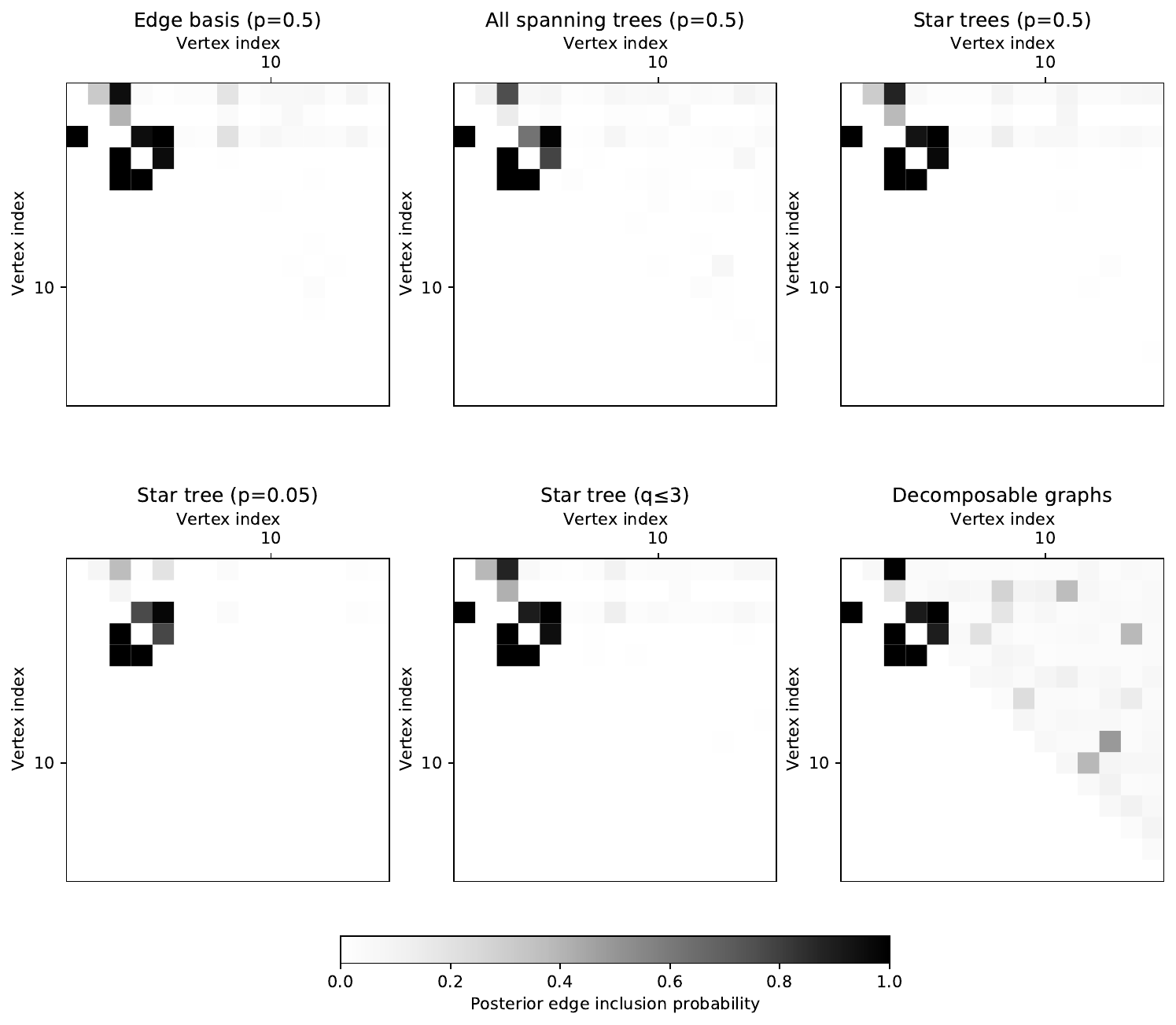}}
	\caption{\addd{The posterior edge inclusion probabilities (upper triangle) and the median probability graph (lower triangle) resulting from the six different priors considered for the simulated data.}}	\label{fig:simul}
}\end{figure}

We compute the posterior for the first five priors using Algorithm~\ref{supp:alg:MCMC} in Supplementary Material where we repeat Step~\ref{supp:step:update_G} nine times for each Step~\ref{supp:step:update_T}.
Step~\ref{supp:step:update_T} involves sampling from $p(T)$.
For the uniform prior over all spanning trees (the second prior),
we use the algorithm from \cite{aldous1990random} to sample from $p(T)$. Schild~\cite{Schild2018} provides a faster alternative.
We set the total number of MCMC steps, i.e.\ executions of Algorithm~\ref{supp:alg:MCMC},
to $10^5$ 
of which $10^4$ 
as burn-in.
Posterior computation for the prior on decomposable graphs is as described in Section~\ref{sec5}.
Figure~\ref{fig:simul} visualises the resulting posterior edge inclusion probabilities.
The probabilities are not notably different between the edge basis and the uniform prior over the cycle space $\mathcal{C}_n$ induced by considering all spanning trees or all star trees with $p=0.5$.

The median probability graphs of the first three priors plus the star tree prior with the number of basis elements $q$ capped at three ($q\leq 3$) have identical recovery of the true underlying graph with the cycle $(3,4,5)$ and edge $(1,3)$ correctly detected,
failure to detect $(1,2)$ and $(2,3)$,
and no false positives.
The extra prior regularisation in the fourth prior with $p=0.05$ results in edge $(1,3)$ not being detected.
The truncation $q\leq 3$ in the fifth prior results in lower inclusion probabilities:
an average probability of $0.054$ compared to $0.056$ for the third prior with $p=0.5$.
The fourth prior with $p=0.05$ yields yet a lower average probability at $0.033$. The prior on decomposable graphs produces an average of 0.11 with notably non-zero probabilities also for edges with endpoints outside of the first five nodes.}

\section{Application to gene expression data}
\label{sec5}

In this section, we restrict the graph space $\mathcal{G}_n$ to the cycle space $\mathcal{C}_n$ while inferring a gene expression network.
The data are gene expression profiles taken from tumours of breast cancer patients obtained from the \texttt{R} package \texttt{gRbase} \citep{dethlefsen2005common} preprocessed following the procedure in \cite{hojsgaard2012graphical} which yields $n=93$ genes of interest from $N=250$ patients.
For a further description of the data collection, see \cite{miller2005expression}. 

For this application, we consider \add{three priors on $\mathcal{G}_n$}: the \add{uniform prior on $\mathcal{G}_n$ corresponding to the} standard ``edge'' basis\add{, the uniform prior on $\mathcal{C}_n$ such as arising from a}
cycle basis
\add{(e.g.\ Proposition~\ref{prop:uniform}) and a uniform prior over all decomposable graphs.}
We use a Metropolis-Hastings algorithm for posterior inference where the proposal is to flip the presence of one basis element in $G$ \add{where we use the edge basis also for the prior on decomposable graphs}.
For the cycle \add{basis}, this proposal is guaranteed to be in $\mathcal{C}_n$ as a vector space is closed under addition.
See Section~\ref{supp:sec:mcmc_uniform} of Supplementary Material for details of the MCMC.
\addd{The computational cost is similar for the edge and cycle basis. Since the Metropolis-Hastings proposal that we use for decomposable graphs is not constrained to decomposable graphs, the rejection rate of that MCMC is high. We therefore run it for longer than the MCMC with the edge and cycle basis.}

\begin{figure}[tbp]
    \centerline{\includegraphics[width=\textwidth]{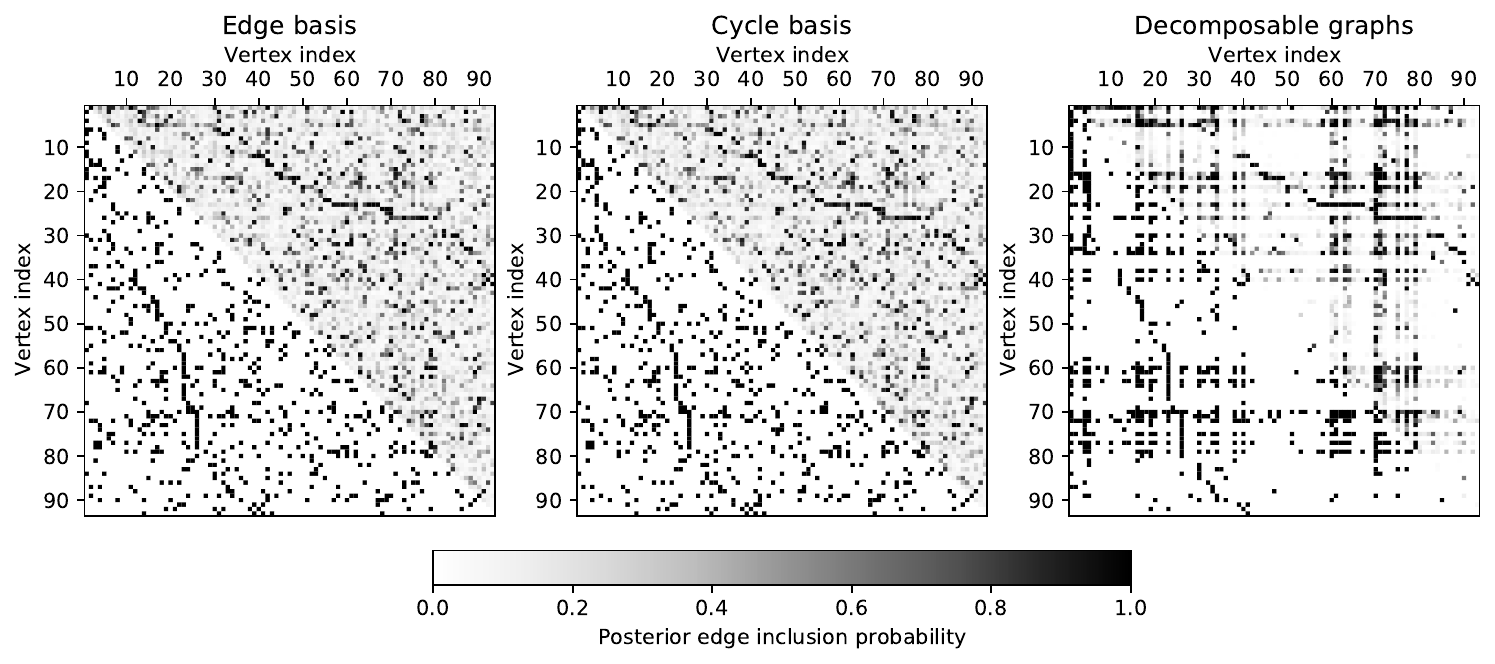}}
	\caption{The posterior edge inclusion probabilities (upper triangle) and the median probability graph (lower triangle) resulting from the \add{standard edge basis (left), cycle basis (centre) and the set of decomposable graphs (right)}.}
	\label{fig: breastcancer-comparison}
\end{figure}

\begin{figure}[tbp]
    \centerline{\includegraphics[width=.95\textwidth]{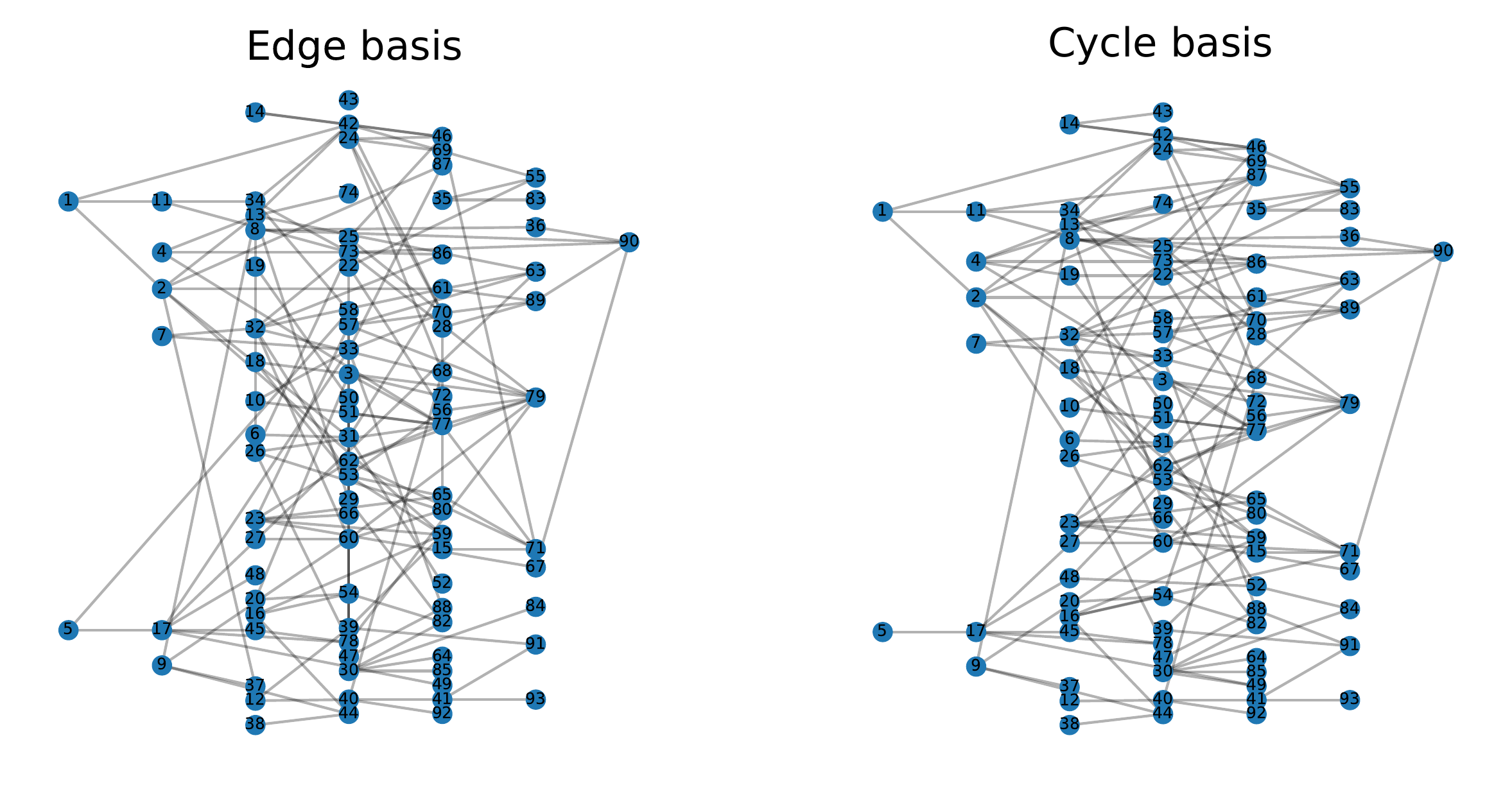}}
	\caption{Graphs obtained by including edges whose posterior inclusion probability is larger than $0.95$ when using the standard edge basis (left) and the cycle basis (right). \addd{The left graph has 141 edges and the right graph has 147 edges, 123 of which appear in both graphs.}}
	\label{fig: breastcancer-graph}
\end{figure}

\begin{table}[tbp]
\begin{center}
\add{
\caption{Minimum, average and maximum absolute difference in posterior edge inclusion probabilities resulting from 
the cycle basis and the set of decomposable graphs relative to the standard edge basis.}
\label{tab:breastcancer-comparison}
    \begin{tabular}{rccc}
    \hline
    & Minimum & Average & Maximum \\
    \hline
    Cycle basis & 0.00 & 0.04 & 0.33 \\
    Decomposable graphs & 0.00 & 0.19 & 1.00 \\
    \hline
    \end{tabular}
}
\end{center}
\end{table}

Figures~\ref{fig: breastcancer-comparison}
and \ref{fig: breastcancer-graph}\add{, and Table~\ref{tab:breastcancer-comparison}}
summarise the resulting graph inference.
There is no \add{major} 
difference between the posterior edge inclusion probabilities with the edge basis and the cycle basis.
This is despite the cycle space restriction with the cycle basis.
These results suggest that restricting inference to the cycle space $\mathcal{C}_n \subset \mathcal{G}_n$
does not substantially affect posterior inference from when the graph space $\mathcal{G}_n$ is not restricted.
\add{This contrasts with the results for decomposable graphs which are substantially different from both the edge and cycle basis results, reflecting that decomposability is a more severe restriction than the cycle space.}

\section{Discussion} \label{sec6}
In this paper, we introduce a generalisation of the edge inclusion prior based on vector spaces and,
in particular, investigate the cycle subspace. 
We present theoretical results about the cycle space and their bases.
While the results presented in Section~\add{\ref{sec2}} are not our own except for Proposition~\ref{thm:dense},
to the extent of our knowledge, this is the first time the idea of restricting inference to the cycle space is introduced in the graphical model literature.

We also study a novel prior based on assigning independent prior basis inclusion probabilities, proving
\add{its} 
degree distribution and shrinkage properties.
We implement an MCMC algorithm that samples from the cycle space. \add{Our algorithm is 
more straightforward to implement when compared to methods for decomposable graphs \citep{giudici1999decomposable,green2013sampling,Stingo2015},} by the fact that vector spaces are closed under addition.
\add{We show empirically} that studying a smaller but dense subset of the graph space does not significantly affect \add{inference of} posterior edge inclusion \add{probabilities}.

Moreover, this paper opens up an opportunity to various extensions of existing methodologies by considering alternative graph vector spaces such as induced by cycle bases as opposed to the edge basis. 
For example, the size-based prior of Armstrong et al.~\cite{armstrong2009bayesian} can be used to shrink to the number of basis elements.
Also, a birth-death proposal on the basis elements can be used instead of the proposal used in this paper that randomly switches one basis element selected from a uniform distribution on all elements.

Lastly, in this paper, we only considered standard GGMs. However, it may be possible to extend this method to other types of graphical models such as multiple graphs \citep{peterson2015bayesian, tan2017bayesian}, Gaussian copulas \citep{dobra2011copula, mohammadi2015bdgraph} and chain graphs \citep{sonntag2015chain, Lu2020}. 






\fund 
\noindent This work was supported by the Singapore Ministry of Education Academic Research Fund Tier~2 under Grant MOE2019-T2-2-100.

\competing 
\noindent There were no competing interests to declare which arose during the preparation or publication process of this article.

\data 
\noindent The data related to the application in Section~\ref{sec5} can be found at
\url{https://github.com/kristoforusbryant/cbmcmc/blob/main/breastcancer/data/breastcancer_data_93_250.csv}.

\supp \noindent The supplementary material for this article can be found at http://doi.org/10.1017/[TO BE SET]. 

%
%
%
%






\bibliographystyle{APT}
\bibliography{bibliography}

\begin{thebibliography}{10}

\bibitem{aldous1990random}
{\sc Aldous, D.~J.} (1990).
\newblock The random walk construction of uniform spanning trees and uniform
  labelled trees.
\newblock {\em SIAM J. Discrete Math.\/} {\bf 3,} 450--465.

\bibitem{armstrong2009bayesian}
{\sc Armstrong, H., Carter, C.~K., Wong, K. F.~K. and Kohn, R.} (2009).
\newblock Bayesian covariance matrix estimation using a mixture of decomposable
  graphical models.
\newblock {\em Statist. Comput.\/} {\bf 19,} 303--316.

\bibitem{atay2005monte}
{\sc Atay-Kayis, A. and Massam, H.} (2005).
\newblock A {M}onte {C}arlo method for computing the marginal likelihood in
  nondecomposable {G}aussian graphical models.
\newblock {\em Biometrika\/} {\bf 92,} 317--335.

\bibitem{sheldon2015linear}
{\sc Axler, S.} (2015).
\newblock {\em Linear Algebra Done Right} 3rd~ed.
\newblock Undergraduate Texts in Mathematics. Springer Cham, Heidelberg.
\newblock Theorem~2.33.

\bibitem{bender1985almost}
{\sc Bender, E., Richmond, L. and Wormald, N.} (1985).
\newblock Almost all chordal graphs split.
\newblock {\em J. Austral. Math. Soc.\/} {\bf 38,} 214--221.

\bibitem{Cayley1889}
{\sc Cayley, A.} (1889).
\newblock A theorem on trees.
\newblock {\em Q. J. Pure Appl. Math.\/} {\bf 23,} 376--378.

\bibitem{chaiken1978matrix}
{\sc Chaiken, S. and Kleitman, D.~J.} (1978).
\newblock Matrix tree theorems.
\newblock {\em J. Comb. Theory Ser. A\/} {\bf 24,} 377--381.

\bibitem{chung1994chordal}
{\sc Chung, F.~R. and Mumford, D.} (1994).
\newblock Chordal completions of planar graphs.
\newblock {\em J. Comb. Theory Ser. B\/} {\bf 62,} 96--106.

\bibitem{dawid1993hyper}
{\sc Dawid, A.~P. and Lauritzen, S.~L.} (1993).
\newblock Hyper {M}arkov laws in the statistical analysis of decomposable
  graphical models.
\newblock {\em Ann. Statist.\/} {\bf 21,} 1272--1317.

\bibitem{dempster1972covariance}
{\sc Dempster, A.~P.} (1972).
\newblock Covariance selection.
\newblock {\em Biometrics\/} {\bf 28,} 157--175.

\bibitem{dethlefsen2005common}
{\sc Dethlefsen, C. and H{\o}jsgaard, S.} (2005).
\newblock A common platform for graphical models in {R}: The {gRbase} package.
\newblock {\em J. Statist. Softw.\/} {\bf 14,} 1--12.

\bibitem{dobra2011copula}
{\sc Dobra, A. and Lenkoski, A.} (2011).
\newblock Copula {G}aussian graphical models and their application to modeling
  functional disability data.
\newblock {\em Ann. Appl. Statist.\/} {\bf 5,} 969--993.

\bibitem{duan2021bayesian}
{\sc Duan, L.~L. and Dunson, D.~B.}
\newblock Bayesian spanning tree: Estimating the backbone of the dependence
  graph 2021.
\newblock {arXiv:2106.16120v1}.

\bibitem{giudici1996learning}
{\sc Giudici, P.} (1996).
\newblock Learning in graphical {G}aussian models.
\newblock In {\em Bayesian Statistics 5}. ed. J.~M. Bernardo, J.~O. Berger,
  A.~P. Dawid, and A.~F.~M. Smith.
\newblock Oxford University Press, Oxford.
\newblock pp.~621--628.

\bibitem{giudici1999decomposable}
{\sc Giudici, P. and Green, P.~J.} (1999).
\newblock Decomposable graphical {G}aussian model determination.
\newblock {\em Biometrika\/} {\bf 86,} 785--801.

\bibitem{green2013sampling}
{\sc Green, P.~J. and Thomas, A.} (2013).
\newblock Sampling decomposable graphs using a {M}arkov chain on junction
  trees.
\newblock {\em Biometrika\/} {\bf 100,} 91--110.

\bibitem{hojsgaard2012graphical}
{\sc H{\o}jsgaard, S., Edwards, D. and Lauritzen, S.} (2012).
\newblock {\em Graphical Models with R}.
\newblock Springer, New York.

\bibitem{jones2005experiments}
{\sc Jones, B., Carvalho, C., Dobra, A., Hans, C., Carter, C. and West, M.}
  (2005).
\newblock Experiments in stochastic computation for high-dimensional graphical
  models.
\newblock {\em Statist. Sci.\/} {\bf 20,} 388--400.

\bibitem{kruskal1956shortest}
{\sc Kruskal, J.~B.} (1956).
\newblock On the shortest spanning subtree of a graph and the traveling
  salesman problem.
\newblock {\em Proc. Amer. Math. Soc.\/} {\bf 7,} 48--50.

\bibitem{lauritzen1996graphical}
{\sc Lauritzen, S.~L.} (1996).
\newblock {\em Graphical Models}.
\newblock Oxford Statistical Science Series. The Clarendon Press, Oxford
  University Press, New York.

\bibitem{lenkoski2011computational}
{\sc Lenkoski, A. and Dobra, A.} (2011).
\newblock Computational aspects related to inference in {G}aussian graphical
  models with the {$G$-Wishart} prior.
\newblock {\em J. Comput. Graph. Statist.\/} {\bf 20,} 140--157.

\bibitem{Lu2020}
{\sc Lu, D., De~Iorio, M., Jasra, A. and Rosner, G.~L.} (2020).
\newblock Bayesian inference for latent chain graphs.
\newblock {\em Found. Data Sci.\/} {\bf 2,} 35--54.

\bibitem{miller2005expression}
{\sc Miller, L.~D., Smeds, J., George, J., Vega, V.~B., Vergara, L., Ploner,
  A., Pawitan, Y., Hall, P., Klaar, S., Liu, E.~T. and Bergh, J.} (2005).
\newblock An expression signature for p53 status in human breast cancer
  predicts mutation status, transcriptional effects, and patient survival.
\newblock {\em Proc. Natl. Acad. Sci. U.S.A.\/} {\bf 102,} 13550--13555.

\bibitem{mohammadi2015bdgraph}
{\sc Mohammadi, R. and Wit, E.~C.} (2019).
\newblock {BDgraph}: An {R} package for {B}ayesian structure learning in
  graphical models.
\newblock {\em J. Statist. Softw.\/} {\bf 89,} 1--30.

\bibitem{murray2012mcmc}
{\sc Murray, I., Ghahramani, Z. and MacKay, D. J.~C.} (2006).
\newblock {MCMC} for doubly-intractable distributions.
\newblock In {\em Proceedings of the Twenty-Second Conference on Uncertainty in
  Artificial Intelligence}.
\newblock UAI’06.
\newblock AUAI Press, Arlington, Virginia, USA.
\newblock p.~359–366.

\bibitem{niu2021bayesian}
{\sc Niu, Y., Pati, D. and Mallick, B.~K.} (2021).
\newblock Bayesian graph selection consistency under model misspecification.
\newblock {\em Bernoulli\/} {\bf 27,} 637--672.

\bibitem{peterson2015bayesian}
{\sc Peterson, C., Stingo, F.~C. and Vannucci, M.} (2015).
\newblock Bayesian inference of multiple {G}aussian graphical models.
\newblock {\em J. Am. Statist. Assoc.\/} {\bf 110,} 159--174.

\bibitem{roverato2002hyper}
{\sc Roverato, A.} (2002).
\newblock Hyper inverse {W}ishart distribution for non-decomposable graphs and
  its application to {B}ayesian inference for {G}aussian graphical models.
\newblock {\em Scand. J. Statist.\/} {\bf 29,} 391--411.

\bibitem{Schild2018}
{\sc Schild, A.} (2018).
\newblock An almost-linear time algorithm for uniform random spanning tree
  generation.
\newblock In {\em Proceedings of 50th Annual {ACM} {SIGACT} Symposium on the
  Theory of Computing ({STOC}'18)}.
\newblock pp.~214--227.

\bibitem{Schwaller2019}
{\sc Schwaller, L., Robin, S. and Stumpf, M.} (2019).
\newblock Closed-form {B}ayesian inference of graphical model structures by
  averaging over trees.
\newblock {\em Journal de la Société Française de Statistique\/} {\bf 160,}
  1--23.

\bibitem{scott2008feature}
{\sc Scott, J.~G. and Carvalho, C.~M.} (2008).
\newblock Feature-inclusion stochastic search for gaussian graphical models.
\newblock {\em J. Comput. Graph. Statist.\/} {\bf 17,} 790--808.

\bibitem{sonntag2015chain}
{\sc Sonntag, D. and Pe{\~n}a, J.~M.} (2015).
\newblock Chain graphs and gene networks.
\newblock In {\em Foundations of Biomedical Knowledge Representation}.
\newblock SpringerNature, Switzerland pp.~159--178.

\bibitem{Stingo2015}
{\sc Stingo, F. and Marchetti, G.~M.} (2015).
\newblock Efficient local updates for undirected graphical models.
\newblock {\em Statist. Comput.\/} {\bf 25,} 159--171.

\bibitem{tan2017bayesian}
{\sc Tan, L.~S., Jasra, A., De~Iorio, M. and Ebbels, T.~M.} (2017).
\newblock Bayesian inference for multiple {G}aussian graphical models with
  application to metabolic association networks.
\newblock {\em Ann. Appl. Statist.\/} {\bf 11,} 2222--2251.

\bibitem{uhler2018exact}
{\sc Uhler, C., Lenkoski, A. and Richards, D.} (2018).
\newblock Exact formulas for the normalizing constants of {W}ishart
  distributions for graphical models.
\newblock {\em Ann. Statist.\/} {\bf 46,} 90--118.

\bibitem{van2021unbiased}
{\sc van~den Boom, W., Jasra, A., De~Iorio, M., Beskos, A. and Eriksson, J.~G.}
  (2022).
\newblock Unbiased approximation of posteriors via coupled particle {M}arkov
  chain {M}onte {C}arlo.
\newblock {\em Statist. Comput.\/} {\bf 32,} 36.

\bibitem{Veblen1912}
{\sc Veblen, O.} (1912).
\newblock An application of modular equations in analysis situs.
\newblock {\em Ann. of Math.\/} {\bf 14,} 86--94.

\bibitem{wallis2010beginner}
{\sc Wallis, W.~D.} (2010).
\newblock {\em A Beginner's Guide to Graph Theory}.
\newblock Springer Science+Business Media, New York.

\bibitem{wang2010simulation}
{\sc Wang, H. and Carvalho, C.~M.} (2010).
\newblock Simulation of hyper-inverse {W}ishart distributions for
  non-decomposable graphs.
\newblock {\em Electron. J. Statist.\/} {\bf 4,} 1470--1475.

\bibitem{wang2012efficient}
{\sc Wang, H. and Li, S.~Z.} (2012).
\newblock Efficient {G}aussian graphical model determination under
  {$G$-Wishart} prior distributions.
\newblock {\em Electron. J. Statist.\/} {\bf 6,} 168--198.

\end{thebibliography}


\begin{thebibliography}{8}
\expandafter\ifx\csname natexlab\endcsname\relax\def\natexlab#1{#1}\fi

\bibitem[{Chaiken \& Kleitman(1978)}]{chaiken1978matrix}
\textsc{Chaiken, S.} \& \textsc{Kleitman, D.~J.} (1978).
\newblock Matrix tree theorems.
\newblock \textit{J. Comb. Theory Ser. A} \textbf{24}, 377--381.

\bibitem[{Duan \& Dunson(2021)}]{duan2021bayesian}
\textsc{Duan, L.~L.} \& \textsc{Dunson, D.~B.} (2021).
\newblock Bayesian spanning tree: Estimating the backbone of the dependence
  graph.
\newblock {arXiv:2106.16120v1}.

\bibitem[{Greenhill et~al.(2017)Greenhill, Isaev, Kwan \&
  McKay}]{greenhill2017average}
\textsc{Greenhill, C.}, \textsc{Isaev, M.}, \textsc{Kwan, M.} \& \textsc{McKay,
  B.~D.} (2017).
\newblock The average number of spanning trees in sparse graphs with given
  degrees.
\newblock \textit{European J. of Combin.} \textbf{63}, 6--25.

\bibitem[{H{\o}jsgaard et~al.(2012)H{\o}jsgaard, Edwards \&
  Lauritzen}]{hojsgaard2012graphical}
\textsc{H{\o}jsgaard, S.}, \textsc{Edwards, D.} \& \textsc{Lauritzen, S.}
  (2012).
\newblock \textit{Graphical Models with R}.
\newblock Springer, New York.

\bibitem[{Kapoor \& Ramesh(2000)}]{kapoor2000algorithm}
\textsc{Kapoor, S.} \& \textsc{Ramesh, H.} (2000).
\newblock An algorithm for enumerating all spanning trees of a directed graph.
\newblock \textit{Algorithmica} \textbf{27}, 120--130.

\bibitem[{Lenkoski \& Dobra(2011)}]{lenkoski2011computational}
\textsc{Lenkoski, A.} \& \textsc{Dobra, A.} (2011).
\newblock Computational aspects related to inference in {G}aussian graphical
  models with the {$G$-Wishart} prior.
\newblock \textit{J. Comput. Graph. Statist.} \textbf{20}, 140--157.

\bibitem[{Mohammadi et~al.(2021)Mohammadi, Massam \& Letac}]{Mohammadi2021}
\textsc{Mohammadi, R.}, \textsc{Massam, H.} \& \textsc{Letac, G.} (2021).
\newblock Accelerating {B}ayesian structure learning in sparse {G}aussian
  graphical models.
\newblock \textit{J. Am. Statist. Assoc.} \textbf{Advance online publication}.

\bibitem[{Schwaller et~al.(2019)Schwaller, Robin \& Stumpf}]{Schwaller2019}
\textsc{Schwaller, L.}, \textsc{Robin, S.} \& \textsc{Stumpf, M.} (2019).
\newblock Closed-form {B}ayesian inference of graphical model structures by
  averaging over trees.
\newblock \textit{Journal de la Société Française de Statistique}
  \textbf{160}, 1--23.

\end{thebibliography}

\end{document}


\title{Supplementary material for ``On a wider class of prior distribution for graphical models''}
\author{Abhinav Natarajan, Willem van den Boom, \\
Kristoforus Bryant Odang and Maria De Iorio}
\date{}
\maketitle

\section{Proofs for propositions from the main text}

\begin{proof}[Proof of Proposition~\ref{main:thm:dense}]
Note that the sum of degrees over every vertex of a graph is even.
Hence, $k$ is an even number, else the sum of the degrees would be odd.
Number the odd-degree vertices as $v_1, v_2, \dots, v_k$.
Now, add (using $\oplus$) edges $(v_i, v_{k/2 + i})$, $i = 1,2 \dots k/2$, to $G$ to obtain $G_C$.
By construction,
$G_C$
has only even degree vertices
such that $G_C\in\mathcal{C}_n$ by Property~\ref{main:p:veblen}.
Thus, there exists a graph in $\mathcal{C}_n$ that differs from $G$ by $k/2$ edges.

We prove the second part of the proposition by contradiction.
Assume that there exists
a graph $G_C\in\mathcal{C}_n$
that differs from $G$ in $l < k/2$ edges.
Flipping these $l$ edges in $G$ to obtain $G_C$
will affect the degree of at most $2l$ vertices since each edge involves two vertices. Thus, $G_C$ shares at least $k - 2l > 0$ odd-degree vertices with $G$
which
contradicts $G_C\in\mathcal{C}_n$
by Property~\ref{main:p:veblen}.
\end{proof}

\begin{proof}[Proof of Proposition~\ref{main:prop:uniform}]
Since the cycle space is a vector space with basis $C = \{c_1, \ldots, c_r\}$, each element of the cycle space is in bijective correspondence with a binary sequence of length $r$. Each binary element of this sequence corresponds to the inclusion of the respective cycle basis element. As the cycle inclusions are independent with probability $0.5$, every binary sequence is equally probable.
The marginal edge inclusion probability of $0.5$ follows from Proposition~\ref{main:prop:edge_prop} with $p=0.5$.
\end{proof}

\begin{proof}[Proof of Proposition~\ref{main:prop:edge_prop}]
Denote $\{c_1, \ldots, c_r\}$ by
$C_e$.
The edge $e$ is included if and only if an odd number of cycles from $C_e$ are included. Suppose $C_k \subset C_e$ is a subset with the odd number of elements $|C_k| = 2k+1$. Then, the probability that the cycles in $C_k$ are included and the cycles in $C_e \setminus C_k$ are not is
\[
p(C_k) = \left(\prod_{c_i \in C_k} p_i \right) \left(\prod_{c_i \in C_e \setminus C_k} (1-p_i)\right).
\]
By summing this probability over all $C_k \subset C_e$, we get the joint probability of edge $e$ and including exactly $2k+1$ cycles as
$\sum_{C_k \subset C_e} p(C_k)$.
Note that this is exactly the coefficient of $x^{2k+1}$ in $f(x)$. Then, summing this coefficient over $k=0,\dots,\lfloor (r-1)/2 \rfloor$ yields the desired result.

\add{
Next, consider the special case where $p=p_1=\dots=p_r$.
Then,
$f(x) = (1-p+px)^r = \sum_{i=0}^r \binom{r}{i} {(1-p)^{r-i}} (px)^i$
where the last equality follows from the binomial theorem.
Thus,
the probability of inclusion of the edge $e$ is
$\sum_{i \text{ odd}} \binom{r}{i} (1-p)^{r-i} p^i$.
By the binomial theorem,
$\sum_{i=0}^r \binom{r}{i} (1-p)^{r-i} p^i = (1 - p + p)^r = 1$
and
$\sum_{i=0}^r \binom{r}{i} (1-p)^{r-i} (-p)^i = (1 - p - p)^r = (1-2p)^r$.
Also,
$\sum_{i=0}^r \binom{r}{i} (1-p)^{r-i} p^i - \sum_{i=0}^r \binom{r}{i} (1-p)^{r-i} (-p)^i
= 2\sum_{i \text{ odd}} \binom{r}{i} (1-p)^{r-i} p^i$
since
\[
    (1-p)^{r-i} p^i - (1-p)^{r-i} (-p)^i = \begin{cases}
        0, & i\text{ even} \\
        2(1-p)^{r-i} p^i, & i\text{ odd}
    \end{cases}
\]
Therefore,
$\sum_{i \text{ odd}} \binom{r}{i} (1-p)^{r-i} p^i = \{1 - (1-2p)^r\}/2$.
}
\end{proof}

\add{
\begin{proof}[Proof of Corollary~\ref{main:cor:edge_prob_star}]
Let $\{v_0, \ldots, v_{n-1}\}$ be the vertices such that the star tree generating the basis $C$ is rooted at $v_0$. We refer to edges of the form $(v_i, v_j)$ where $i, j \neq 0$ as \emph{peripheral edges}. We refer to edges of the form $(v_0, v_i)$ as \emph{rooted edges}.

First, consider a rooted edge $(v_0,v_i)$.
Exactly $(n-2)$ edges in the complement $\overline{T}$
of $T$
have $v_i$ as an endpoint.
Each of these edges correspond to their own unique basis element in $C$.
No other basis elements in $C$ contain $(v_0,v_i)$.
Thus, $r=n-2$ basis elements contain $(v_0,v_i)$.
On the other hand,
each peripheral edge is in $\overline{T}$
and is thus contained in only $r=1$ basis element
due to the bijection between edges in $\overline{T}$
and elements of $C$.
The required result for a given $T$ now follows from Proposition~\ref{main:prop:edge_prop}.

For the result marginalising over the uniform distribution over all star trees,
note that there are $n$ star trees on $n$ vertices.
Each edge is rooted in 2 of these trees and peripheral in the other $(n-2)$. Thus, the marginal edge inclusion probability follows as
\begin{multline*}
    \Pr\{(i,j) \in E\}
    = \\ \Pr\{\text{$(i,j)$ is peripheral}\} \Pr\{(i,j) \in E \mid \text{$(i,j)$ is peripheral}\} \\
    +
    \Pr\{\text{$(i,j)$ is rooted}\} \Pr\{(i,j) \in E \mid \text{$(i,j)$ is rooted}\} \\
    = \frac{n-2}{n} p
    + \frac{2}{n} \{1-(1-2p)^{n-2}\}/2
\end{multline*}
which simplifies to the required result.
\end{proof}
}

\begin{proof}[Proof of Proposition~\ref{main:prop:edge_prop_joint}]
Using the independence of the cycle inclusions, it is seen that for any $I \subset \{1, \ldots, r\}$ the joint inclusion probability of the cycles $C_I = \{c_i : i \in I\}$ is given by the coefficient of $\prod_{i \in I} t_i$ in $f$. The inclusion of $C_I$ leads to the inclusion of the edges $\{e_{a(i)}, e_{b(i)} : i \in I\}$, notwithstanding cancellations due to intersecting cycles. The probability generating function $h \in \RR[x_1, \ldots, x_m]$ for the edge inclusions (with numerosity and ignoring binary cancellations) is therefore given by replacing each $t_i$ in $f$ with $x_{a(i)}x_{b(i)}$. Binary cancellations are then handled by reducing all exponents modulo $2$, which is equivalent to taking the image of $h$ in the quotient ring. 
\end{proof}

\begin{proof}[Proof of Proposition~\ref{main:prop:deg_dist}]
For odd $k$,
$\Pr\{\deg(v_i) = k\} = 0$ by Property~\ref{main:p:veblen}.
The remainder of this proof thus considers even $k$.

Note that each edge $(v_i, v_j)$, $i,j\ne 0$, is present in exactly one cycle basis element whose vertices are $\{v_0, v_i, v_j\}$. Therefore, the inclusion of edge $(v_i, v_j)$ is equivalent to the inclusion of the cycle $\{v_0, v_i, v_j\}$. Fix $i > 0$. If $v_i$ has degree $k > 0$, then it is connected to either (i) $k$ other vertices that are not the root vertex, or (ii) the root vertex and $k-1$ other vertices.
\begin{enumerate}[label=(\roman*)]
    \item Consider the first case. For each of the $k$ vertices connected to $v_i$ via an edge, we must have the inclusion of the corresponding cycle, and for each of the $n-2-k$ vertices not connected to $v_i$ via an edge, the corresponding cycle cannot be included. This happens with probability $p^k(1-p)^{n-2-k}$ where $p$ is the inclusion probability for any cycle basis element. The inclusion of exactly $k$ cycles containing the edge $(v_0, v_i)$ ensures that this edge is not included since $k$ is even. Since there are $\binom{n-2}{k}$ ways to pick $k$ non-root vertices that are connected to $v_i$, the probability of $v_i$ being connected to $k$ non-root vertices is $\binom{n-2}{k}p^k (1-p)^{n-2-k}$.
    \item In the second case, by a similar argument to the above, there are exactly $k-1$ cycles of the form $\{v_0, v_i, v_j\}$ that are included.
    This happens with probability $\binom{n-2}{k-1} p^{k-1}(1-p)^{n-k-1}$. 
    Since $k$ is even, this results in the edge $\{v_0, v_i\}$ being included.
\end{enumerate}
Adding the probabilities of (i) and (ii) yields the required result for $k \geq 2$.
For $k=0$, only (i) is possible such that
$\Pr\{\deg(v_i) = k\} = \binom{n-2}{k} p^k(1-p)^{n-2-k} = (1-p)^{n-2}$.
\end{proof}

\begin{proof}[Proof of Proposition~\ref{main:prop: shrinkage}]
\add{Recall the definition of peripheral and rooted edges from the proof of
Corollary~\ref{main:cor:edge_prob_star}.}
The lower bound $|E|\geq \add{q}$ follows from the observation that
there are exactly $\add{q}$ peripheral edges included in $G$.
The remainder of this proof thus considers the upper bound.

The maximum of three edges per cycle ($|E|=3\add{q}$) is attained when there are no intersections between the included cycles. Including a cycle $\{v_0, v_i, v_j\}$ corresponds to choosing a pair $\{v_i, v_j\}$ from the $(n-1)$ non-root vertices.
Avoiding intersections between included cycles corresponds to choosing all such pairs disjoint.
By the pigeonhole principle, there can be at most $m$ mutually disjoint pairs.
Thus,
$|E|=3\add{q}$ is attainable if and only if
$\add{q} \leq m$.
Also,
$|E|\leq 3\add{q}$ for $\add{q} \leq m$.

By Property~\ref{main:p:veblen}, the degree of $v_0$ and thus the number of rooted edges in $G$ must be even.
The largest even number less than the number $(n-1)$ of non-root vertices is $2m$.
Additionally, there are $\add{q}$ peripheral edges.
Thus,
$|E|\leq \add{q} + 2m$.
\end{proof}

\section{Edge union of spanning trees}
\label{sec:spanning_tree}

We explore the prior induced by edge unions of spanning trees in this section.
As seen in \add{\citet{hojsgaard2012graphical}, \citet{Schwaller2019} and \citet{duan2021bayesian}}, the spanning tree structure has proved relevant in graphical model inference.
We thus consider defining a graph prior $p(G \mid k)$ where
the graph $G=(V, E)$ is an edge union of $k$ spanning trees,
i.e.\
$E = \cup_{i = 1}^k T_i$
where $T_i$ is the set of edges of a spanning tree.
We assume that the inclusion of each spanning tree is equally likely and independent a priori.
Then, to infer the induced posterior distribution of the graphical model using Metropolis-Hastings, the following ratio needs to be computed:
$$\frac{p(E' =  \cup_{i = 1}^{k'} T_i\mid k')}{p(E = \cup_{i = 1}^k T_i\mid k)} = \frac{|Y(G', k')| / \binom{|\tau(G')|}{k'}}{|Y(G, k)| / \binom{|\tau(G)|}{k}}$$
where
the prime marks proposed values in the MCMC,
$|Y(G, k)|$ is the number of ways to write $G$ as the edge union of $k$ distinct spanning trees, $k' \in \{k - 1, k + 1\}$, and $|\tau(G)|$ is the number of all spanning trees of $G$. 
Note that by Kirchoff's matrix tree theorem (Theorem~\ref{thm: kmtt} \add{on page~\pageref{thm: kmtt}}), $|\tau(G)|$ can be computed in $O(n^3)$ time. Hence, the term left to be computed is the ratio $|Y(G', k')|/|Y(G, k)|$.

The remainder of this section is structured as follows.
Section~\ref{sec:notation} introduces notation.
Then,
Sections~\ref{sec:direct}, \ref{sec:lower} and \ref{sec:upper}
compute, lower bound and upper bound the ratio $|Y(G', k')|/|Y(G, k)|$, respectively.

\subsection{Notation}
\label{sec:notation}

In this Section~\ref{sec:spanning_tree}, all graphs considered are connected.
Let $K\in \mathcal{G}_n$ be the complete graph on the vertex set $V$. For any graph $G$, denote the set of spanning trees of $G$ by $\tau(G)$. For any $e\in E$, denote the set of spanning trees of $G$ that contain $e$ by $\tau_e(G)\subset \tau(G)$. Let $\binom{\tau(K)}{k}$ be the set with the $\binom{|\tau(K)|}{k}$ subsets of $\tau(K)$ of size $k$. There is a map 
\[
\Phi_k: \binom{\tau(K)}{k} \to \mathcal{G}_n
\]
given by taking the edge union of the $k$ trees in each element of $\binom{\tau(K)}{k}$.
For $G \in \mathrm{im}(\Phi_k) = \{\Phi_k(\mathcal{T}) : \mathcal{T} \in  \binom{\tau(K)}{k}\}$,
define
the set of combinations of $k$ spanning trees that yield $G$ as
$Y(G, k) = \Phi_k^{-1}(G) = \{\mathcal{T} : \Phi_k(\mathcal{T}) = G \}$.
Let $G_1 + G_2$ denote the graph obtained by taking the edge union of graphs $G_1$ and $G_2$.
Our main aim is to estimate the ratio
\begin{equation}\label{ratio}
\frac{|Y(G', k')|}{|Y(G, k)|} =
\frac{|Y(G+T_0, k+1)|}{|Y(G, k)|}
\end{equation}
for a fixed $G \in \mathrm{im}(\Phi_k)$ and $T_0 \in \tau(K)$.
Note that here we define $G' = G+T_0$
and $k' = k + 1$.

\subsection{Direct approach}
\label{sec:direct}

One approach to estimate this ratio is to compute $|Y(G, k)|$ and $|Y(G', k')|$ separately. 
First, consider the brute force strategy that enumerates every tree in $G$. 
This is infeasible given how the number of spanning trees $|\tau(G)|$ grows super-exponentially with $m = |E|$ \citep{greenhill2017average} and the enumeration algorithm has complexity $O(n + m + |\tau(G)|)$
\citep{kapoor2000algorithm}.
In this section, we show an algorithm that improves on the super-exponential brute-force algorithm by counting the number of ways to write $G$ as a union of its spanning trees in $O(2^m\, n^{3})$ time. 

\begin{lemma} \label{lem:1}
Every edge union of $k$ spanning trees of $G$ is either equal to $G$ or is an edge-induced subgraph of $G$ that spans $V$ and is connected.
\end{lemma}

\begin{proof}
This follows by contradiction since
the negation would mean that the spanning tree union contains an edge not in $G$. 
\end{proof}

\begin{proposition} \label{prop:recursive}
Let
$G_C$ denote any connected graph with as set of edges a proper subset of the edge set $E$ of $G$. Then,
\begin{align}
  |Y(G,k)| = \binom{|\tau(G)|}{k} - \sum_{G_C} |Y(G_C,k)|
\end{align}
if $|\tau(G)| \geq k$ and $|Y(G,k)| = 0$ otherwise. 
\end{proposition}

\begin{proof}
Note that $\binom{|\tau(G)|}{k}$ is the number of possible combinations of $k$ distinct spanning trees of $G$. Further, by Lemma~\ref{lem:1}, any union of spanning trees of $G$ that is connected and not equal to $G$ is a connected edge-induced subset of $G$.
\end{proof}

Hence, having proved the recursive formula in (1), we analyse the complexity of its calculation. \add{We recall the following proposition due to Kirchoff (see \citeauthor{chaiken1978matrix} (\citeyear{chaiken1978matrix}) for a reference).}
\begin{theorem}[Matrix tree theorem]\label{thm: kmtt}
The number $|\tau(G)|$ of spanning trees of $G$ is equal to any cofactor of the Laplacian matrix of $G$.
\end{theorem}
This theorem implies that counting $\tau(G)$ is equivalent to finding the determinant of an $(n-1) \times (n-1)$ matrix which has complexity $O(n^3)$. 

Now, we prove that we can compute $|Y(G,k)|$ in $O(2^m\, n^{3})$ time. 

\begin{proposition}
$|Y(G,k)|$ can be calculated in $O(2^m\, n^{3})$ time. 
\end{proposition}

\begin{proof}
Consider the recursive scheme implied by Proposition~\ref{prop:recursive}.
Since every connected edge-induced subgraph of every $G_C$ is also an edge-induced subgraph of $G$, using dynamic programming and memoisation, we only need to calculate every unique $|Y(G_C, k)|$ of which there are at most $2^m$. Calculating $|\tau(G)|$ at each step involves calculating a determinant, per Theorem~\ref{thm: kmtt}, which is $O(n^3)$. Hence, the overall complexity is $O(2^m\, n^3)$.
\end{proof}

Since $G$ is constrained to be the union of only $k$ spanning trees with each $(n-1)$ edges, $m \leq k (n - 1)$. Hence, the algorithm is $O(2^{kn}\, n^3)$.

\subsection{Lower bound}
\label{sec:lower}

In this section, we derive a lower bound for \eqref{ratio}.
Consider the equation
\begin{equation}\label{eqn1}
    G' = G + T_0 = G + T
\end{equation} 
for unknown $T \in \tau(K)$ and known $T_0$. Let the set of solutions $T$ of \eqref{eqn1} be $\lambda(G, T_0)$. It is clear that 
\[
Y(G+T_0, k+1) \supset \left\{S \cup \{T\} : S \in Y(G, k),\, T \in \lambda(G, T_0) \right\}
\]
and hence 
\[
|Y(G+T_0, k+1)| \geq |Y(G, k)| |\lambda(G, T_0)| \implies \frac{|Y(G+T_0, k+1)|}{|Y(G, k)|} \geq |\lambda(G, T_0)|.
\]
We therefore start by computing $|\lambda(G, T_0)|$. Consider the equivalence relation $\sim$ on $V$ given by
two edges being connected
in $G'-G$ (whose vertices are $V$ and edges are $E'\setminus E$). Let $G'' = (\mathcal{V}, \mathcal{E})$ be the quotient graph $G' / \sim$. Let $\pi :G' \to G''$ denote the quotient map. 

\begin{lemma}
Suppose $C$ is a simple cycle of $G'$. Then, $\pi(C)$ is either a single point, a single edge connecting a pair of vertices, or a simple cycle of $G''$.
\end{lemma}

\begin{lemma}\label{lambda is weighted tree enumerator}
Suppose 
\[
|\lambda(G, T_0)| = \sum_{\mathcal{T} \in \tau(G'')}\prod_{e \in \mathcal{E}(\mathcal{T})} |\pi^{-1}(e)|
\]
where $\mathcal{E}(\mathcal{T})$ is the edge set of $\mathcal{T}$.
\end{lemma}
\begin{proof}
Let $\mathcal{T} \in \tau(G'')$. We construct a spanning tree $T \in \lambda(G, T_0)$ as follows. For each edge $e = (v, w) \in \mathcal{T}$ choose a representative $(v', w') \in \pi^{-1}(e)$ and set $(v', w')$ to be an edge in $T$. We then add in all the edges of $G'-G$. Clearly the resulting graph satisfies \eqref{eqn1} because it contains every edge of $G'-G$. It is spanning and connected because it connects every connected component of $G'-G$ and contains every edge in $G'-G$. We show that it is acyclic. Suppose $C$ is a cycle of $T$; without loss of generality we can assume it is simple. Note that $G' - G$ must be acyclic because its edge set is a subset of edges of the tree $T_0$, so $C$ cannot be a subgraph of $G'-G$. Therefore $\pi(C)$ is not a single vertex. If $\pi(C)$ is a single edge, then $C$ contains two distinct edges joining two components of $G'-G$, which is not possible by our construction of $T$. Therefore by the previous lemma $\pi(C)$ is a simple cycle of $\mathcal{T}$ which is a contradiction. 

Every $T \in \lambda(G, T_0)$ can be constructed in this manner. Since $T$ is a tree it must descend to a tree $\mathcal{T}$ in $G''$. An edge $(v, w) \in T$ corresponds to the edge $([v], [w]) \in \mathcal{T}$; by choosing $v$ and $w$ as the representatives in the construction above we get back the edge $(v, w)$. 

It is also clear that distinct trees $\mathcal{T}_1$ and $\mathcal{T}_2$ in $H$ will give different trees in $G'$. For a single $\mathcal{T}$, each distinct choice of representatives for each edge will give a distinct tree $T$. Therefore $|\lambda(G, T_0)|$ is the number of ways to carry out the above construction, which is given by the formula above.
\end{proof}

We now give weights to the edges of $G''$: let each edge $e \in \mathcal{E}(G'')$ have the weight $|\pi^{-1}(e)|$. The weight of a tree $\mathcal{T} \in \tau(G'')$ is defined as the product of the weight of its edges:
\[
w(\mathcal{T}) = \prod_{e \in \mathcal{E}(\mathcal{T})} w(e).
\]
Then 
\[
\lambda(G, T_0) = \sum_{\mathcal{T} \in \tau(G'')} w(\mathcal{T}).
\]
The expression above is called the weighted tree enumerator of $G''$. The weighted Kirchoff matrix tree theorem gives a way to compute this quantity.

\begin{theorem}[Formula for $\lambda(G, T_0)$]
Let $L$ be the Laplacian of $G''$ taking into account the weights of the edges, i.e.
\[
L = D - A
\]
where $A$ is the weighted adjancency matrix and $D$ is the diagonal matrix whose $i$th entry is the sum of weights of edges incident to vertex $i$. Then $|\lambda(G, T_0)|$ is the absolute value of any cofactor of $L$. 
\end{theorem}

\begin{example}
If $G$ and $G'$ have the same edges then $G'-G$ has no edges and $G = G' = G''$. In this case $\lambda(G,T_0)$ is the number of trees of $G$ and this agrees with the formula above (the empty product is defined as 1). 
\end{example}

\subsection{Upper bound}
\label{sec:upper}

In this section, we derive an upper bound for \eqref{ratio}
by bounding $|\lambda(G, T_0)|$.
Suppose $G= (V,E) \in \mathrm{im}(\Phi_k)$, so that for $1 \leq j \leq k$ there exists $T_j \in \tau(K)$ such that $E = \bigcup_{j=1}^k \mathcal{E}(T_j)$. Define $G_j = (V, E_j)$ by $E_j=\bigcup_{i=1}^{j} \mathcal{E}(T_i)$. We therefore have
\[
G = T_k + G_{k-1}.
\]
Using the results from above we get
\[
|Y(G, k)| \geq |Y(G_{k-1}, k-1)|\, |\lambda(G_{k-1}, T_k )|.
\]
By iterating we get 
\[
    |Y(G, k)| \geq \prod_{j=1}^{k-1} |\lambda(G_j, T_{j+1})|.
\]
Now note that the ordering of the $T_j$ does not matter to the argument. Let $S_k$ denote the $k$th permutation group and for $\sigma \in S_k$ let $G_{j\sigma}$ denote $\bigcup_{1}^j T_{\sigma(i)}$. Then
\begin{equation}\label{eqn2}
    |Y(G, k)| \geq \max_{\sigma \in S_k}\prod_{j=1}^{k-1} |\lambda(G_{j\sigma}, T_{\sigma(j+1)})|.
\end{equation}
This bounds the denominator in \eqref{ratio}. Next we address the numerator. Suppose $S \in Y(G', k+1)$. For each edge $e \in G'$, at least one of the trees in $S$ must contain $e$. Therefore,
\[
Y(G', k+1) \subset \left\{ \{T\} \cup S: T \in \tau_e(G'), S \in \binom{\tau(K)}{k}\right\}.
\]
This gives the inequality
\begin{equation}\label{eqn3}
|Y(G', k+1)| \leq \min_{e \in E(G')} |\tau_e(G')|\, |\tau(K)|^{k}.
\end{equation}
The computation of $\tau_e(G')$ can be carried out as follows: let $G/e$ be the graph obtained by contracting the edge $e$, let $\pi$ be the quotient, and for any edge $e'$ in $G/e$ assign the weight $\pi^{-1}(e')$ to $e'$. Then $|\tau_e(G)|$ is simply the weighted tree enumerator of $G/e$; the proof is similar to the proof for Lemma~\ref{lambda is weighted tree enumerator}. 
Equations~\ref{eqn2} and \ref{eqn3} \noeqref{eqn2} \noeqref{eqn3}
together give an upper bound:
\begin{equation}
    \frac{|Y(G', k+1)|}{|Y(G, k)|} \leq \frac{\min_{e \in E(G')}|\tau_e(G')| |\tau(K)|^{k}}{\max_{\sigma \in S_k}\prod_{j=1}^{k-1} |\lambda(G_{j\sigma}, T_{\sigma(j+1)})|}.
\end{equation}
Altogether, we have:
\begin{theorem}
\[
|\lambda(G, T_0)| \leq \frac{|Y(G', k+1)|}{|Y(G, k)|} \leq \frac{\min_{e \in E(G')}|\tau_e(G')| |\tau(K)|^{k}}{\max_{\sigma \in S_k}\prod_{j=1}^{k-1} |\lambda(G_{j\sigma}, T_{\sigma(j+1)})|}.
\]
\end{theorem}

Figure~\ref{fig:bound} shows that the lower and upper bounds derived here for $|Y(G', k+1)|/|Y(G, k)|$
can be multiple orders of magnitude different from $|Y(G', k+1)|/|Y(G, k)|$.
As such, they do not provide usable approximations for a Metropolis-Hastings acceptance ratio.

\begin{figure}[htbp]
    \centerline{\includegraphics[width=\textwidth]{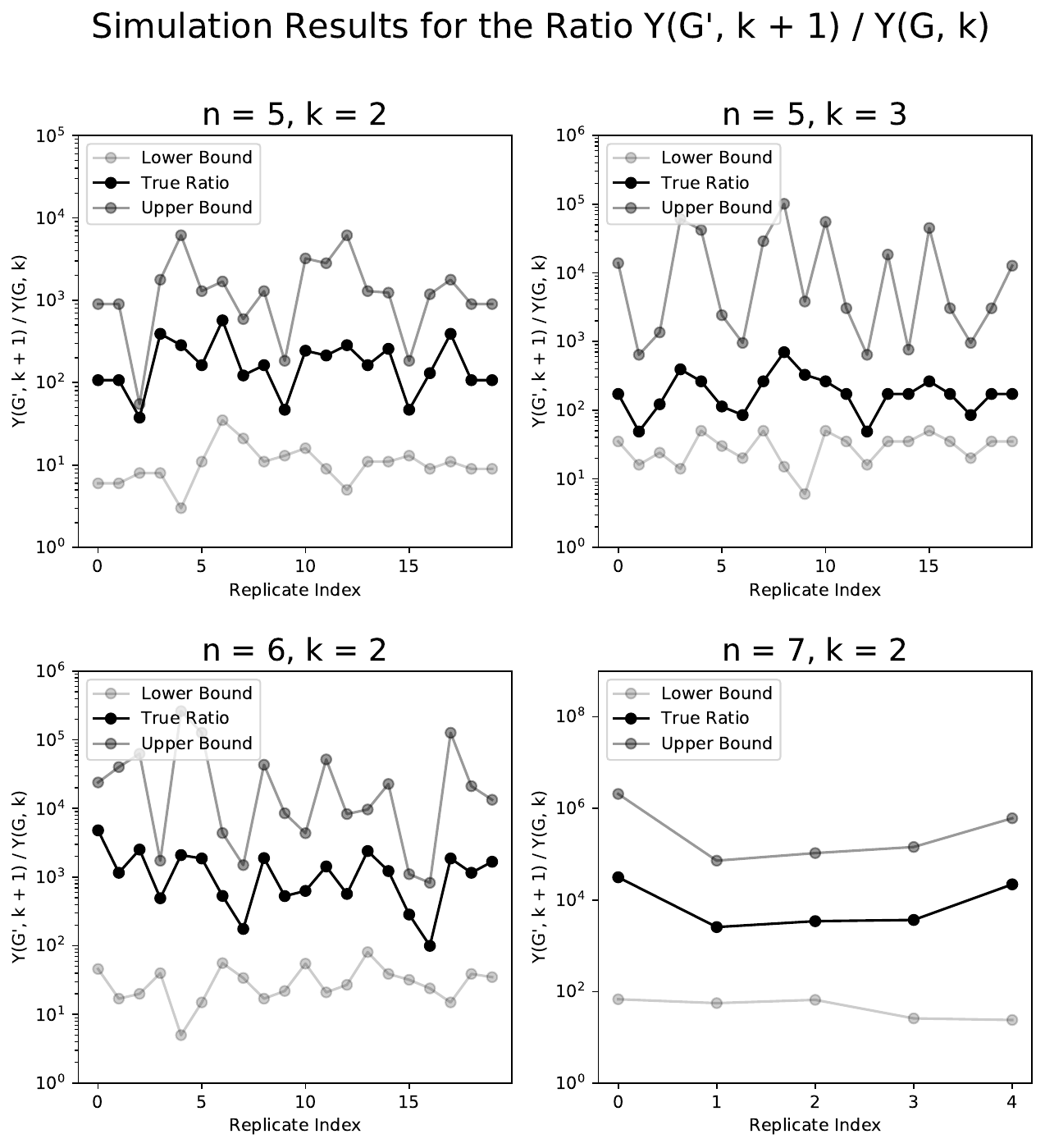}}
	\caption{The value, lower bound from Section~\ref{sec:lower} and upper bound from Section~\ref{sec:upper} of
	$|Y(G', k+1)|/|Y(G, k)|$
	on a logarithmic scale
	for fixed number of vertices $n$ and number of spanning trees $k$.
	Here, the replicates consider randomly sampled $G$ and $G'$ where $G$ is a union of $k$ uniformly sampled spanning trees
	and $G'=G+T_0$ with $T_0$ a uniformly sampled spanning tree.
 }	\label{fig:bound}
\end{figure}

\section{Posterior inference on the cycle space using MCMC}
\label{sec:mcmc}

\add{We describe two different MCMC algorithms for posterior inference on the cycle space.
Firstly, we consider general posteriors arising from prior distributions on spanning trees and associated cycle bases.
Secondly, we consider the uniform prior over the cycle space which allows for computationally less expensive MCMC.}

\subsection{General spanning tree prior $p(T)$}
\label{sec:mcmc_general}

\add{We consider the following prior construction for the cycle space $\mathcal{C}_n$.
Note that the distribution over $\mathcal{C}_n$ can change with the basis $C$ used: see for instance Corollary~\ref{main:cor:edge_prob_star}.
Thus, we do not constrain our inference to a single $C$.
Let $p(T)$ be the prior distribution on spanning trees which each correspond to a $C$.
Then, a prior distribution on basis element inclusions induces a prior $p(G\mid T)$ over $\mathcal{C}_n$.
The joint posterior of $G$ and $T$ follows then from the likelihood $p(X\mid G)$ in \eqref{main:eq: marginal-lik} as
$p(G,T\mid X)\propto p(T)\, p(G\mid T)\, p(X\mid G)$.

We use Metropolis-Hastings within Gibbs to sample from $p(G,T\mid X)$.
Specifically, we update the graph $G$ using a Metropolis proposal $q(G'\mid G,T)$
that involves flipping one random basis element in $G$, i.e.\ flipping all edges in one uniformly sampled element of $C$.
Thus, this is a multiple-edge update.
Also, the proposal is guaranteed to be in $\mathcal{C}_n$ because the latter is a vector space.
Note that the support of $q(G'\mid G,T)$ varies with $T$
which determines the basis elements $C$ to choose from.
In other words, which graphs are in the ``neighbourhood'' of $G$ varies with $T$,
resulting in differing exploration of the graph space depending on the current value of $T$.

For the spanning tree $T$, we use its prior distribution as an independent Metropolis-Hastings proposal, $q(T'\mid T) = p(T')$.
A change in $T$ corresponds to a change in basis $C$
and thus requires the decomposition of $G$ in terms of this new basis.
This decomposition involves a matrix inverse over the field $\mathbb{Z}_2$
which is computationally expensive for a large number of vertices $n$ as the size of the matrix to invert is super-exponential in $n$.
Therefore,
we repeat the update of $G$
multiple times for each update of $T$.
Algorithm~\ref{alg:MCMC}
summarises the resulting MCMC algorithm.

}

\begin{algorithm}
\add{
\caption{MCMC step for $p(G,T\mid X)$ for a general spanning tree prior $p(T)$ \label{alg:MCMC}}
\begin{enumerate}
\item \label{step:update_G}
Perform the Metropolis update for $G$ a fixed number of times:
\begin{enumerate}
    \item
    Sample a proposed $G'$ from $q(G'\mid G,T)$.
    \item \label{step:lik_ratio}
    Set $G = G'$ with probability $\min(1,\alpha_G)$ where
    \[
        \alpha_G = \frac{p(G',T\mid X)\, q(G\mid G',T)}{p(G,T\mid X)\, q(G'\mid G,T)}
        = \frac{p(G'\mid T)\, p(X\mid G')}{p(G\mid T)\, p(X\mid G)}
    \]
\end{enumerate}
\item \label{step:update_T}
Perform the Metropolis-Hastings update for $T$:
\begin{enumerate}
    \item
    Sample a proposed $T'$ from $q(T'\mid T) = p(T')$.
    \item \label{step:decomposition}
    Compute the decomposition of $G$ in terms of the basis generated by $T'$
    for the evaluation of $p(G\mid T')$.
    \item
    Set $T'=T$ with probability $\min(1,\alpha_T)$ where
    \[
        \alpha_T = \frac{p(G,T'\mid X)\, q(T\mid T')}{p(G,T\mid X)\, q(T'\mid T)}
        = \frac{p(G\mid T')}{p(G\mid T)}
    \]
\end{enumerate}
\end{enumerate}
}
\end{algorithm}

To compute $\frac{p(X \mid G')}{p(X \mid G)} = \frac{I_{G'}(\delta + N, D + U) }{I_G(\delta + N, D + U)} \frac{ I_G(\delta, D)}{ I_{G'}(\delta, D)}$
in \add{Step~\ref{step:lik_ratio}} of Algorithm~\ref{alg:MCMC}, 
we use two different approximations. For $I_{G'}(\delta + N, D + U)$ and $I_{G}(\delta + N, D + U)$, we use the Laplace approximation from \citet{lenkoski2011computational}. For $\frac{ I_G(\delta, D)}{ I_{G'}(\delta, D)}$, we use the approximation for the ratio of $G$-Wishart normalising constants from \citet{Mohammadi2021}. This is because the Laplace approximation is known to be inaccurate when the degrees of freedom $\delta$ is small \citep{lenkoski2011computational}. In our case, $G$ and $G'$ differ by at least three edges
while \citet{Mohammadi2021} consider single edge updates.
We therefore approximate
$\frac{ I_G(\delta, D)}{ I_{G'}(\delta, D)}$
as the product of a sequence of (at least three) ratio approximations
from \citet{Mohammadi2021}
each involving a single edge change.

\add{
\subsection{Uniform prior over the cycle space}
\label{sec:mcmc_uniform}
}

We now discuss the case with
$p(G\mid T)$
being the uniform prior over $\mathcal{C}_n$
as it allows for simplifications.
\add{Then, the choice of spanning tree $T$ is irrelevant per Proposition~\ref{main:prop:uniform}, negating the need for the expensive decomposition in Step~\ref{step:decomposition} in Algorithm~\ref{alg:MCMC}.}
Also, sampling \add{a cycle of three vertices}
uniformly at random and flipping the corresponding edges in $G$ as proposal
yields $p(G'\mid G)=p(G\mid G')$.
Note that $p(G')/p(G)=1$ by construction.
Thus, the Metropolis acceptance ratio follows as $\alpha = p(X\mid G') / p(X\mid G)$.
We use
this MCMC algorithm that proposes to change a uniformly sampled 3-cycle
in Section~\ref{main:sec5}.

\add{We use this MCMC algorithm also for the uniform prior over all graphs $\mathcal{G}_n$
and over all decomposable graphs
with the modification that the proposal $p(G'\mid G)$
follows from randomly flipping one instead of three edges
for the results in Section~\ref{main:sec5}.
We run this algorithm
for $10^6$
iterations with $10^5$
burn-in iterations
for $\mathcal{C}_n$ and $\mathcal{G}_n$,
and for $10^7$
iterations with $10^6$
burn-in iterations
for the decomposable graphs:
we use more MCMC iterations for the decomposable graphs since many Metropolis proposals $G'$ resulting from flipping an edge are rejected because they are not decomposable, and thus have no prior nor posterior mass, resulting in worse mixing of the MCMC chain.}
Implementation of the MCMC algorithms
and code used to generate the results in Section~\ref{main:sec5}
are available from \url{https://github.com/kristoforusbryant/cbmcmc}.

\bibliographystyle{biometrika}
\bibliography{bibliography}